\documentclass[a4paper]{article}

\usepackage{a4wide}
\usepackage[UKenglish]{babel}
\usepackage{graphicx,pgfplots,subfigure}
	\pgfplotsset{compat=newest}
	\graphicspath{{./}{figure/}{figure/hydro/}}
\usepackage[colorlinks=true,linkcolor=blue,citecolor=red]{hyperref}
\usepackage{amssymb,amsthm,mathtools,bbold}
\usepackage[inline]{enumitem}

\newtheorem{theorem}{Theorem}[section]

\newtheorem{lemma}[theorem]{Lemma}
\theoremstyle{remark}\newtheorem{remark}[theorem]{Remark}

\newcommand{\abs}[1]{\left\lvert#1\right\rvert}
\DeclarePairedDelimiter\ave{\langle}{\rangle}
\newcommand{\cL}{\mathcal{L}}
\newcommand{\E}{\mathbb{E}}
\newcommand{\cI}{\mathcal{I}}
\newcommand{\cH}{\mathcal{H}}
\newcommand{\pr}[1]{{}^\prime\!#1}
\newcommand{\R}{\mathbb{R}}
\newcommand{\ucontr}{u}%{\mathfrak{u}}
\newcommand{\Var}{\operatorname{Var}}

\allowdisplaybreaks

\begin{document}
\title{Model-based assessment of the impact of driver-assist vehicles using kinetic theory}

\author{Benedetto Piccoli\thanks{Department of Mathematical Sciences, Rutgers University - Camden, Camden NJ, USA} \and
		Andrea Tosin\thanks{Department of Mathematical Sciences ``G. L. Lagrange'', Politecnico di Torino, Torino, Italy} \and
		Mattia Zanella\thanks{Department of Mathematics ``F. Casorati'', University of Pavia, Pavia, Italy}
		}
\date{}

\maketitle

\begin{abstract}
In this paper we consider a kinetic description of follow-the-leader traffic models, which we use to study the effect of vehicle-wise driver-assist control strategies at various scales, from that of the local traffic up to that of the macroscopic stream of vehicles. We provide a theoretical evidence of the fact that some typical control strategies, such as the alignment of the speeds and the optimisation of the time headways, impact on the local traffic features (for instance, the speed and headway dispersion responsible for local traffic instabilities) but have virtually no effect on the observable macroscopic traffic trends (for instance, the flux/throughput of vehicles). This unobvious conclusion, which is in very nice agreement with recent field studies on autonomous vehicles, suggests that the kinetic approach may be a valid tool for an organic multiscale investigation and possibly design of driver-assist algorithms.

\medskip


\end{abstract}

\section{Introduction}
The last decade has been characterised by an increasing level of automation in decision processes, thanks to a deep technological development and algorithmic sensing. Nowadays, the coexistence of human and computer-based assistance is one of the main goals in the context of smart cities and vehicular dynamics, due to its potential ability to mitigate dangerous practices and to dissipate congestions created by the behavioural responses of the drivers, see e.g.~\cite{hoogendoorn2014TRR,riostorres2017IEEE}. As an example, we mention the Advanced Driver-Assistance Systems (ADAS), which constitute a real interface between human drivers and machine-based decision making. 

Since motor vehicles are commonly durable goods, decisions by policy makers in this market are expected to impact over a time horizon of several years. On the other hand, a full penetration of the new technologies is unrealistic in the near future, even in those markets with high consumer demand, see e.g.~\cite{schoettle2015TR}. Therefore, it is of interest to study transient regimes, in which the penetration rate of driver-assist vehicles is quite small, say in the benchmark range $5\%$ to $10\%$. Remarkably enough, field experiments recently confirmed the enhancement produced by suitable ADAS-type protocols in the regularisation of the stream of vehicles, even for small penetration rates. At the same time, virtually no observable effect was measured at the level of vehicle flux (throughput in the engineering literature), which seems to remain unchanged also after the activation of automatic control procedures, see~\cite{stern2018TRC}. 

In order to get a deeper understanding and mastery of these issues, it is fundamental to complement the experimental observations with theoretical insights into the links between the elementary vehicle-wise dynamics and their observable effects at larger scales. It is therefore of interest to set up model-based procedures able to assess the impact of ADAS-type technologies at different scales, starting from the control of vehicle-to-vehicle interactions up to the quantification of the resulting aggregate effect. In this paper we propose to ground such a multiscale analysis on the sound mathematical-physical framework of the kinetic theory, which has now a quite consolidated tradition also in the modelling of vehicular traffic. Far from claiming to be exhaustive, we recall here the pioneering works~\cite{paveri1975TR,prigogine1971BOOK} together with modern kinetic methods~\cite{delitala2007M3AS,herty2010KRM,illner2003CMS,klar1997JSP} and recent advancements~\cite{dimarco2019JSP-b,herty2018SIAP,tosin2018IFAC,tosin2019MMS,visconti2017MMS}. Taking advantage of a statistical approach, kinetic-type equations naturally link the microscopic dynamics of few interacting agents representative of a much larger group to the observable macroscopic trends of the system as a whole. Therefore, they provide an effective framework to pass from the scale of single vehicles, where behavioural ``forces'' and ADAS-type controls are defined, to consistent hydrodynamic equations describing the aggregate traffic flow under the action of implementable vehicle-wise controls. Such an organic multiscale analysis of cruise control strategies is the hallmark of this work compared to other recent approaches, which either focus on single-scale descriptions, see e.g.~\cite{delis2015CMA,delis2018TRR,ntousakis2015TRP,stern2018TRC}, or propose a heuristic embedding of point particles, representing the driver-assist vehicles, in a continuous flow of standard vehicles, see e.g.~\cite{dellemonache2019HAL,garavello2019ARXIV}. It is worth mentioning that control problems in connection with kinetic equations are an active research line, which has been recently investigated in the context of self-organisation and games to reproduce competitive scenarios and to force the emergence of patterns, see~\cite{albi2015CMS,albi2014PTRSA,albi2019JSP}. Moreover, links with mean field optimal control problems and performance bounds have also been investigated, see~\cite{albi2017AMO,herty2017DCDS}.

As usual in the kinetic approach, the first step consists in defining the microscopic model of the interactions among the vehicles. Rather than postulating it from scratch, here we take advantage of the celebrated Follow-the-Leader (FTL) class of microscopic traffic models~\cite{gazis1961OR}, which, thanks to their simplicity and versatility, had a strong impact on the transportation engineering community. An FTL traffic model describes one-directional dynamics of a platoon of vehicles, assuming that overtaking is not possible and that, at each time, a vehicle regulates its speed only on the basis of the relative speed and distance from its leading vehicle. This description implies pairwise interactions, which are particularly suited to a kinetic approach. Next, we complement the FTL dynamics of a few vehicles with a control term designed in such a way to induce the controlled vehicles to keep a recommended distance from their leading vehicles while smoothing the speed gaps. Such control actions are motivated by the general idea of dampening traffic instabilities, such as e.g. stop-and-go waves, which are indeed produced by inhomogeneous distributions of the headways (viz. the distance from the leading vehicle) and the speeds in the traffic stream. We point out that, thanks to underlying FTL model, controlling the headway turns out to be equivalent to controlling the time headway, which is actually more customary in the transportation engineering practice. The resulting control problem can be explicitly solved for a generic pair of interacting vehicles, taking into account probabilistically that one of them may or may not be controlled. The solution is a totally decentralised binary feedback control, which can be straightforwardly embedded in a Boltzmann-type kinetic equation with localised interactions. Then we are in a position to investigate the distributed impact of the aforementioned control strategies on the whole population of vehicles. In particular, we prove that these strategies reduce the local dispersion of both the headway and the speed distribution and we also quantify such a reduction in terms of the penetration rate of ADAS-type vehicles and the level of traffic congestion.

As a final step, we investigate the large scale, viz. macroscopic, impact of the introduced binary control strategies by means of a suitable hydrodynamic limit based on the relaxation of the system towards the local equilibrium distribution of the kinetic equation (the ``Maxwellian'' in the jargon of the classical kinetic theory). It is worth pointing out that the derivation of hydrodynamic equations in the non-classical framework of multi-agent systems, see e.g.~\cite{carrillo2017BOOKCH,degond2008M3AS,duering2007PHYSA,pareschi2019JNS,shu2019ARXIV}, is a still underexplored topic, due to the general lack of information about the large time trends of such systems. In this work we show that such an information may actually be recovered for the system at hand, at least in a particular regime of the parameters of the microscopic interactions. On this basis, we provide an organic method to derive macroscopic traffic equations consistent with microscopic models of vehicle-wise control. In particular, the macroscopic model that we obtain is a conservation law for the traffic density, whose flux is directly obtained from the aforementioned relaxation towards the local kinetic equilibrium and is explicitly parametrised by the penetration rate of the ADAS-type vehicles. This allows us to prove that the flux/throughput exhibits indeed a very mild dependence on the penetration rate, thereby providing a theoretical support to the experimental observations mentioned at the beginning.

In more detail, the paper is organised as follows. In Section~\ref{sect:FTL_micro} we introduce the general class of FTL microscopic dynamics and we derive therefrom a consistent binary interaction scheme expressing the variation of the headway of a vehicle due to an interaction with the leading vehicle. In Section~\ref{sect:control} we embed in the FTL-inspired binary dynamics a vehicle-wise control, designed in such a way to enforce a safety distance and align the speed of a vehicle to that of the leading vehicle. In Section~\ref{sect:boltzmann.hydro} we consider an inhomogeneous Boltzmann-type kinetic description of the controlled binary interaction dynamics, that we use to study the impact of the vehicle-wise control on some local traffic features such as the headway, speed and time headway distributions as well as the headway and speed variance. Moreover, in a suitable hydrodynamic regime we deduce a first order macroscopic traffic model incorporating the action of the introduced vehicle-wise control. In Section~\ref{sect:numerics} we report several numerical tests, which show the consistency of our hierarchical approach. Finally, in Section~\ref{sect:conclusions} we summarise the main results of the paper and we briefly sketch possible follow-ups.

\section{Follow-the-Leader-inspired binary interactions}\label{sect:FTL_micro}
In their general formulation, Follow-the-Leader (FTL) traffic models can be written as:
\begin{equation}
	\begin{cases}
		\dot{x}_i=v_i \\
		\dot{v}_i=\dfrac{av_i^m}{\left(x_{i+1}-x_i\right)^n}\left(v_{i+1}-v_i\right),
	\end{cases}
	\label{eq:FTL}
\end{equation}
cf.~\cite{gazis1961OR}, where $x_i,\,x_{i+1}\in\R$ are the positions of two consecutive vehicles in the traffic stream, $v_i,\,v_{i+1}\geq 0$ are their speeds and the factor $av_i^m/\left(x_{i+1}-x_i\right)^n$, with $a>0$ and $m,\,n\in\mathbb{N}$, represents the sensitivity of the $i$th driver, i.e. his/her promptness in adapting the speed to that of the leading vehicle.

Denoting by $s_i:=x_{i+1}-x_{i}$ the \textit{headway} between the $i$th and $(i+1)$th vehicles, from~\eqref{eq:FTL} we deduce
\begin{equation}
	\frac{\dot{v}_i}{v_i^m}=a\frac{\dot{s}_i}{s_i^n}
	\label{eq:FTL.v-s}
\end{equation}
whence, assuming $m,\,n>1$,
\begin{equation}
	v_i=\frac{s_i^{\frac{n-1}{m-1}}}{{\left(a\frac{m-1}{n-1}+Cs_i^{n-1}\right)}^\frac{1}{m-1}},
	\label{eq:v-s_general}
\end{equation}
where $C\geq 0$ is an arbitrary integration constant. Notice that $v_i$ is essentially proportional to $s_i^{(n-1)/(m-1)}$ at small headway ($s_i\to 0^+$), namely in dense traffic, while it approaches the value ${(1/C)}^{1/{(m-1)}}$ from below at large headway ($s_i\to +\infty$), namely in free traffic. Therefore, if we assume that $v_i=1$ represents, as usual, the maximum non-dimensional speed of a vehicle, a meaningful choice, which we will henceforth make, is $C=1$.

\begin{figure}[!t]
\centering
\resizebox{11cm}{!}{
\begin{tikzpicture}[
	declare function={
		f(\s,\a)=\s/(\a+\s);
	}
]
\begin{axis}[
	xlabel={$s_i$},
	ylabel={$v_i$},
	xmin=0,
	xmax=20,
	ymin=0,
	ymax=1.1,
	legend style={draw=none},
	legend cell align=left,
	legend pos=outer north east,
]
	\draw [style=dashed,thin] (0,1) -- (20,1);
	\addplot[gray,domain=0:20,smooth,very thick,style={solid},samples=100]{f(x,0.5)}; \addlegendentry{$a=0.5$};
	\addplot[blue,domain=0:20,smooth,very thick,style={dashed}]{f(x,1)}; \addlegendentry{$a=1$}
	\addplot[cyan,domain=0:20,smooth,very thick,style={dash dot}]{f(x,2)}; \addlegendentry{$a=2$}
\end{axis}
\end{tikzpicture}
}
\caption{The relationship~\eqref{eq:v-s} between the headway $s_i$ and the speed $v_i$ plotted for various $a$.}
\label{fig:v-s}
\end{figure}
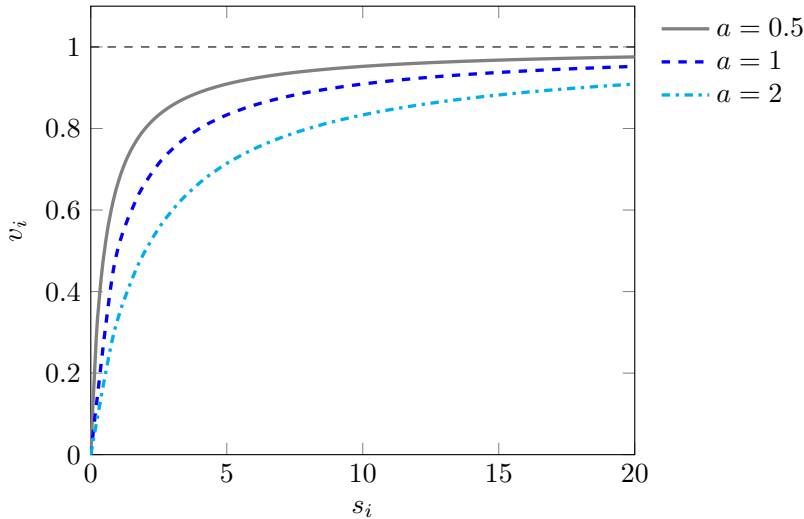

A special class of relationships~\eqref{eq:v-s_general} is obtained for $m=n$:
$$ v_i=\frac{s_i}{{\left(a+s_i^{n-1}\right)}^\frac{1}{n-1}}. $$
whose prototypical case is the one with $n=2$, i.e.
\begin{equation}
	v_i=\frac{s_i}{a+s_i},
	\label{eq:v-s}
\end{equation}
see Figure~\ref{fig:v-s}.

\begin{remark}
Typical FTL models considered in the literature are those obtained from~\eqref{eq:FTL} with $m=0$ and $n=1,\,2$. With these choices of the exponents, instead of~\eqref{eq:v-s} we get from~\eqref{eq:FTL.v-s} the following relationships between the speed and the headway of the vehicles (up to arbitrary integration constants):
$$ v_i=
	\begin{cases}
		a\log{s_i} & \text{if } m=0,\,n=1 \\[1mm]
		1-\dfrac{a}{s_i} & \text{if } m=0,\,n=2.
	\end{cases}
$$
We observe that, in both cases, $s_i$ cannot span the whole $\R_+$: the restrictions $s_i\geq 1$ and $s_i\geq a$, respectively, are needed in order for the speed to be non-negative. Moreover, in the first case $v_i$ is also unbounded above for $s_i\in [1,\,+\infty)$. In this paper, we prefer to consider $m=n=2$, because the resulting relationship~\eqref{eq:v-s} allows for $s_i\in\R_+$ with $v_i$ bounded between $0$ and $1$.
\end{remark}

Taking the difference between the $i$th and the $(i+1)$th equations of the form~\eqref{eq:FTL} with $m=n=2$, and using~\eqref{eq:v-s}, we further get the equation for the headway:
$$ \frac{d}{dt}\left[\dot{s}_i-a\left(\frac{1}{a+s_i}-\frac{1}{a+s_{i+1}}\right)\right]=0, $$
which, in order for the \textit{jammed traffic} state, i.e. $s_i(t)=0$ for all $t\geq 0$ and all $i$, to be an admissible solution, we choose to satisfy with
$$ \dot{s}_i=a\left(\frac{1}{a+s_i}-\frac{1}{a+s_{i+1}}\right). $$
Now, writing $s:=s_i(t)$ and $s_\ast:=s_{i+1}(t)$ for a generic pair of interacting vehicles, we approximate the time derivative in a small time interval $\Delta{t}>0$, understood as the \textit{reaction time} of the drivers, with the forward Euler formula, in the same spirit as~\cite{carrillo2010SIMA}. Since for the development of the theory high values of the parameter $a$ will be significant, it is meaningful to take here $\Delta{t}=\frac{1}{a}$. On the whole, we recover the binary interaction rule
$$ s'=s+\frac{1}{a+s}-\frac{1}{a+s_\ast}, $$
where $s':=s_i(t+\Delta{t})$ is the post-interaction headway.

In order to include in this description some stochastic effects linked to uncontrollable aspects of the dynamics, we add to the previous rule a fluctuation of the form $s\eta$, where $\eta$ is a centred random variable such that
$$ \ave{\eta}=0, \qquad \Var(\eta)=\ave{\eta^2}=:\sigma^2>0, $$
where $\ave{\cdot}$ denotes expectation. The coefficient $s$ ensures that the greater the headway the greater the fluctuation, meaning that the more far apart from each other the vehicles are the more their dynamics are dictated by other factors than mutual interactions.

Thus, the interaction rule that we consider is finally
\begin{equation}
	s'=s+\frac{1}{a+s}-\frac{1}{a+s_\ast}+s\eta.
	\label{eq:binary.s}
\end{equation}
In order to be physically admissible, this rule has to be such that $s'\geq 0$ for all $s,\,s_\ast\geq 0$. Since
$$ s'=\left(1-\frac{1}{(a+s)(a+s_\ast)}+\eta\right)s+\frac{s_\ast}{(a+s)(a+s_\ast)}, $$
we see that such a condition is guaranteed if
$$ 1-\frac{1}{(a+s)(a+s_\ast)}+\eta\geq 0, $$
which, owing to $\frac{1}{(a+s)(a+s_\ast)}\leq\frac{1}{a^2}$, is certainly satisfied if $\eta\geq\frac{1}{a^2}-1$. This implies that the stochastic fluctuations admissible in~\eqref{eq:binary.s} must have a support bounded from the left. At the same time, it is necessary to assume $\frac{1}{a^2}-1<0$, i.e. $a>1$ consistently with the anticipated requirement that $a$ be large, for otherwise $\eta$ cannot have zero mean and strictly positive variance.

Finally, considering that vehicle interactions are anisotropic, and in particular that they happen only with the vehicle in front, the leading vehicle $i+1$ will not modify its headway because of the rear vehicle $i$. Hence $s_\ast'=s_\ast$ as far as the pairwise interaction between the vehicles $i,\,i+1$ is concerned.

On the whole, the binary interaction that we consider is
\begin{equation}
	s'=s+\frac{1}{a+s}-\frac{1}{a+s_\ast}+s\eta, \qquad
		s_\ast'=s_\ast,
	\label{eq:binary}
\end{equation}
with $\eta\geq\frac{1}{a^2}-1$ and $a>1$.

\begin{remark} \label{rem:time_headway}
The \textit{time headway} between two consecutive vehicles is defined as the time gap between their passages through a given section of the road, cf.~\cite{HCM2000TRB}. Using~\eqref{eq:v-s}, we deduce that the time headway $\tau_i$ between the vehicles $i$ and $i+1$ is
\begin{equation}
	\tau_i:=\frac{s_i}{v_i}=a+s_i.
	\label{eq:tau}
\end{equation}
Since $s_i\geq 0$, it results $\tau_i\geq a$, hence the parameter $a$ defines the minimum time headway. Recalling that the reaction time of the drivers is $\Delta{t}=\frac{1}{a}$, we see that $a>1$ is the minimal condition necessary to guarantee $\tau_i>\Delta{t}$, in such a way that drivers have enough time to react to their leading vehicles and avoid car accidents (i.e. negative headways).
\end{remark}

\section{Driver-assist control}
\label{sect:control}
Driver-assist controls are implemented at the level of single vehicles with the aim of enhancing the road safety and improving the global flow of traffic, for instance by dissipating stop-and-go waves responsible for sudden speed variations and large fuel consumption~\cite{stern2018TRC}.

In mathematical terms, this problem can be formalised by modifying the interaction rules~\eqref{eq:binary} as follows:
\begin{equation}
	s'=s+\frac{1}{a+s}-\frac{1}{a+s_\ast}+\Theta\ucontr+s\eta, \qquad s_\ast'=s_\ast, 
	\label{eq:binary.u}
\end{equation}
where $\ucontr$ is an instantaneous control applied by the driver-assist device, which will be deduced from the optimisation of a suitably defined functional, and $\Theta\in\{0,\,1\}$ is a random variable expressing the fact that a generic vehicle may or may not be equipped with driver-assist technologies. In particular, we choose $\Theta\sim\operatorname{Bernoulli}(p)$, where $p\in [0,\,1]$, corresponding to the probability that $\Theta=1$, is the so-called \textit{penetration rate}, namely the percentage of vehicles in the traffic stream equipped with the aforesaid technology.

We design the control so that it has a twofold action:
\begin{itemize}
\item on one hand, it tries to induce a vehicle to maintain a \textit{safety distance} from its leading vehicle. We identify such a distance with a prescribed function $\rho\mapsto s_d(\rho)$ of the traffic density, where the subindex $d$ stands for ``desired'' (distance), and we assume that, during a binary interaction, the control $\ucontr$ acts so as to minimise a suitable non-negative function of the difference $s_d(\rho)-s'$;
\item on the other hand, it tries to align a vehicle speed to the one of the leading vehicle, so as to reduce the necessity of sudden speed variations in case of changes in the traffic flow. Considering that, owing to~\eqref{eq:v-s}, there is a one-to-one relationship between the speed of a vehicle and its headway from the leading vehicle, such a speed alignment is straightforwardly pursued by minimising, during a binary interaction, a suitable non-negative function of the difference $s_\ast'-s'$.
\end{itemize}
We formalise these arguments by introducing the following \textit{binary} functional:
\begin{equation}
	J(s',\,s_\ast',\,\ucontr):=\frac{1}{2}\ave*{\mu{(s_d(\rho)-s')}^2+(1-\mu){(s_\ast'-s')}^2+\nu\ucontr^2},
	\label{eq:J}
\end{equation}
where $\mu\in [0,\,1]$ is a coefficient balancing the bias of the control towards either effect discussed before and $\nu>0$ is a penalisation coefficient, which can be interpreted as the cost of the control. The average $\ave{\cdot}$ is taken with respect to the distribution of the stochastic fluctuation $\eta$ in~\eqref{eq:binary.u}.

\begin{remark}
\begin{enumerate*}[label=(\roman*)]
\item Since there is a one-to-one relationship between $s$ and $v$, the first term of the functional $J$ may be possibly interpreted in the light of the optimal velocity models, c.f.~\cite{bando1995PRE}. Here, the equivalent of a density-dependent optimal speed is the density-dependent optimal headway $s_d(\rho)$.
\item Still concerning the first term of the functional $J$, we observe that real autonomous cruise controls normally do not try to align the headway of a vehicle to a prescribed optimal headway but rather control the time headway. This is because headway alignment encounters specific engineering drawbacks due to internal vehicle dynamics and results in instabilities. Nevertheless, in our model vehicles are point particles and there is a one-to-one relationship between the headway and the time headway, cf.~\eqref{eq:tau}. Moreover, we bypass the inner dynamics of the vehicles for modelling purposes and to allow for kinetic approaches, thus for our model it is equivalent to control either the headway or the time headway.
\end{enumerate*}
\end{remark}

The optimal control $\ucontr^\ast$ is chosen as
$$ \ucontr^\ast:=\operatorname*{arg\,min}_{\ucontr\in\mathcal{U}}J(s',\,s_\ast',\,\ucontr) $$
subject to~\eqref{eq:binary.u}, where $\mathcal{U}=\{u\in\R\,:\,s'\geq 0\}$ is the set of the admissible controls, namely those which guarantee that the post-interaction headway $s'$ in~\eqref{eq:binary.u} is physically consistent.

The minimisation of~\eqref{eq:J} can be done e.g., by forming the Lagrangian
$$ \cL(s',\,\ucontr,\,\lambda):=J(s',\,s_\ast,\,\ucontr)+\lambda\ave*{s'-s-\left(\frac{1}{a+s}-\frac{1}{a+s_\ast}+\Theta\ucontr\right)-s\eta}, $$
where $\lambda\in\R$ is the Lagrange multiplier associated with the first constraint in~\eqref{eq:binary.u}. Instead, the second constraint has been directly forced in $J$. Then the optimality conditions read
\begin{equation*}
	\begin{cases}
		\partial_\ucontr\cL=\nu\ucontr-\gamma\Theta\lambda=0 \\[1mm]
		\partial_{s'}\cL=\mu\ave{s'-s_d(\rho)}+(1-\mu)\ave{s'-s_\ast}+\lambda=0 \\[1mm]
		\partial_\lambda\cL=\ave*{s'-s-\left(\dfrac{1}{a+s}-\dfrac{1}{a+s_\ast}+\Theta\ucontr\right)}=0,
	\end{cases}
\end{equation*}
whence, solving with respect to $\ucontr$, we determine the optimal control
\begin{equation}
	\ucontr^\ast=\frac{\Theta}{\nu+\Theta^2}\left[\mu s_d(\rho)+(1-\mu)s_\ast-s\right]
		-\frac{\Theta}{\nu+\Theta^2}\left(\frac{1}{a+s}-\frac{1}{a+s_\ast}\right)
	\label{eq:ustar}
\end{equation}
in \textit{feedback form}, i.e. as a function of the pre-interaction headways $s$, $s_\ast$. Plugging into~\eqref{eq:binary.u} yields finally the controlled binary interactions
\begin{align}
	\begin{aligned}[c]
		s' &= s+\frac{\nu}{\nu+\Theta^2}\left(\dfrac{1}{a+s}-\dfrac{1}{a+s_\ast}\right)
			+\dfrac{\Theta^2}{\nu+\Theta^2}\left[\mu s_d(\rho)+(1-\mu)s_\ast-s\right]+s\eta \\
		s_\ast' &= s_\ast.
	\end{aligned}
	\label{eq:binary.controlled}
\end{align}

As previously anticipated, the admissibility of the control~\eqref{eq:ustar} is linked to the possibility to guarantee $s'\geq 0$ in~\eqref{eq:binary.controlled} for every choice of $s,\,s_\ast\geq 0$. Writing
\begin{align*}
	s' &= \left(1-\frac{\nu}{\left(\nu+\Theta^2\right)(a+s)(a+s_\ast)}-\frac{\Theta^2}{\nu+\Theta^2}+\eta\right)s
		+\frac{\nu}{\nu+\Theta^2}s_\ast \\
	&\phantom{=} +\frac{\Theta^2}{\nu+\Theta^2}[\mu s_d(\rho)+(1-\mu)s_\ast]
\end{align*}
we see that a sufficient condition for $s'\geq 0$ is
$$ \eta\geq\frac{1}{\nu+\Theta^2}\left(\frac{\nu}{(a+s)(a+s_\ast)}+\Theta^2\right)-1, $$
which, since $\frac{1}{(a+s)(a+s_\ast)}\leq\frac{1}{a^2}$ and considering that $\Theta$ takes only the values $0$, $1$, is further enforced by requiring
$$ \eta\geq\frac{1}{\nu}\left(\frac{\nu}{a^2}+1\right)-1. $$
Therefore, the support of the stochastic fluctuation $\eta$ has to be again bounded from the left and, furthermore,
$$ \frac{1}{\nu}\left(\frac{\nu}{a^2}+1\right)-1<0 $$
for consistency with the requirements $\ave{\eta}=0$, $\ave{\eta^2}>0$. This implies
$$ a>1, \qquad \nu>\frac{a^2}{a^2-1}. $$

Notice that the first restriction coincides with the one imposed in the uncontrolled case~\eqref{eq:binary}. The second restriction, instead, forces a non-zero lower bound on the penalisation coefficient $\nu$, hence in practice it asserts that an admissible control cannot be too cheap.

\section{Boltzmann-type kinetic description and hydrodynamics}
\label{sect:boltzmann.hydro}
A statistical description of a traffic stream composed by vehicles interacting according to the microscopic rules discussed in Section~\ref{sect:control} can be obtained by the methods of kinetic theory. To this purpose, we characterise the microscopic state of a generic vehicle by its position $x\in\R$ along the road and its headway $s\in\R_+$ from the leading vehicle. Next, we introduce the distribution function $f=f(t,\,x,\,s)$, which is such that $f(t,\,x,\,s)\,dx\,ds$ gives, at time $t\geq 0$, the fraction of vehicles located in the interval $[x,\,x+dx]$ with a headway from the leading vehicle comprised between $s$ and $s+ds$.

An evolution equation for $f$ is obtained by appealing to the classical principles of the collisional kinetic theory\footnote{Here and henceforth, according to the classical jargon of kinetic theory, the word ``collision'' and its by-products will be occasionally used as synonyms of ``interaction''.}:
\begin{align}
	\begin{aligned}[b]
		\partial_t\int_{\R_+}\varphi(s)f(t,\,x,\,s)\,ds &+ \partial_x\int_{\R_+}\frac{s}{a+s}\varphi(s)f(t,\,x,\,s)\,ds \\
		&= \frac{1}{2}\E_\Theta\left[\int_{\R_+}\int_{\R_+}\ave{\varphi(s')-\varphi(s)}f(t,\,x,\,s)f(t,\,x,\,s_\ast)\,ds\,ds_\ast\right],
	\end{aligned}
	\label{eq:boltzmann}
\end{align}
where $\varphi$ is any test function representing an \textit{observable quantity}, i.e. a quantity which can be computed out of the knowledge of the microscopic state $s$, $\ave{\cdot}$ denotes, as usual, the average with respect to the distribution of the stochastic fluctuation $\eta$ and $\E_\Theta$ is the expectation with respect to the random variable $\Theta$.

On the whole,~\eqref{eq:boltzmann} is the weak form of a \textit{Boltzmann-type} collisional kinetic equation. In fact, the collision term on the right-hand side contains the distribution functions of the interacting vehicles computed in the same space position $x$. This is clearly a local approximation, however compensated -- to some extent -- by the headway $s$, which considers indeed the actual clear space between the vehicles. In $\varphi(s')$, the quantity $s'$ is given by~\eqref{eq:binary.controlled}.

The transport term on the left-hand side of~\eqref{eq:boltzmann} is written taking into account that~\eqref{eq:v-s} implies, for each vehicle, the kinematic relation
$$ \frac{dx}{dt}=\frac{s}{a+s}. $$

In principle, hydrodynamic traffic equations can be obtained from~\eqref{eq:boltzmann} by choosing $\varphi(s)=s^n$, $n\in\mathbb{N}$, so that the first term on the left-hand side gives the time variation of the statistical moments of $f$ with respect to the variable $s$. The latter are indeed the observable macroscopic quantities, for instance:
\begin{itemize}
\item the \textit{density} of the vehicles in the point $x$ at time $t$:
\begin{equation}
	\rho(t,\,x):=\int_{\R_+}f(t,\,x,\,s)\,ds;
	\label{eq:rho}
\end{equation}
\item the \textit{mean headway} among the vehicles in the point $x$ at time $t$:
\begin{equation}
	h(t,\,x):=\frac{1}{\rho(t,\,x)}\int_{\R_+}sf(t,\,x,\,s)\,ds.
	\label{eq:h}
\end{equation}
\end{itemize}
However, proceeding directly in this way leads to a closure issue, because the second term on the left-hand side of~\eqref{eq:boltzmann} cannot be expressed, in full generality, in terms of macroscopic quantities derived from statistical moments of the distribution function $f$.

To circumvent such a difficulty of the theory, it is useful to perform a hyperbolic scaling of time and space:
\begin{equation}
	t\to\frac{2}{\epsilon}t, \qquad x\to\frac{2}{\epsilon}x,
	\label{eq:scaling}
\end{equation}
where $0<\epsilon\ll 1$, which leads one to rewrite~\eqref{eq:boltzmann} in the form
\begin{align}
	\begin{aligned}[b]
		\partial_t\int_{\R_+}\varphi(s)f(t,\,x,\,s)\,ds &+ \partial_x\int_{\R_+}\frac{s}{a+s}\varphi(s)f(t,\,x,\,s)\,ds \\
		&= \frac{1}{\epsilon}\E_\Theta\left[\int_{\R_+}\int_{\R_+}\ave{\varphi(s')-\varphi(s)}f(t,\,x,\,s)
			f(t,\,x,\,s_\ast)\,ds\,ds_\ast\right].
	\end{aligned}
	\label{eq:boltzmann.scaled}
\end{align}

Basically, the scaling~\eqref{eq:scaling} produces the coefficient $1/\epsilon$ in front of the Boltzmann-type interaction term, hence $\epsilon$ plays in this context a role similar to that of the Knudsen number in the classical kinetic theory. Since we are assuming that $\epsilon$ is small, a hydrodynamic regime is justified. Taking inspiration from~\cite{duering2007PHYSA}, such a regime can be described by a splitting of~\eqref{eq:boltzmann.scaled} totally analogous to that often adopted in the numerical solution of the Boltzmann equation, see e.g.~\cite{dimarco2018JCP,pareschi2001SISC}. One first solves the fast interactions:
\begin{equation}
	\partial_t\int_{\R_+}\varphi(s)f(t,\,x,\,s)\,ds=\frac{1}{\epsilon}\E_\Theta\left[\int_{\R_+}\int_{\R_+}
		\ave{\varphi(s')-\varphi(s)}f(t,\,x,\,s)f(t,\,x,\,s_\ast)\,ds\,ds_\ast\right],
	\label{eq:collision_step}
\end{equation}
which, owing to the high frequency $1/\epsilon$, reach quickly an equilibrium described by a local asymptotic distribution function (the analogous of a local Maxwellian in the classical kinetic theory). Next, one transports such a local equilibrium distribution according to the left-hand side of~\eqref{eq:boltzmann.scaled} on the slower (i.e. hydrodynamic) scale:
\begin{equation}
	\partial_t\int_{\R_+}\varphi(s)f(t,\,x,\,s)\,ds+\partial_x\int_{\R_+}\frac{s}{a+s}\varphi(s)f(t,\,x,\,s)\,ds=0.
	\label{eq:transport_step}
\end{equation}

\subsection{Local equilibrium distribution}
The splitting procedure outlined above requires to identify the local Maxwellian from the interaction step~\eqref{eq:collision_step}. Unfortunately, this is in general a hard task for ``collisional'' Boltzmann-type equations such as~\eqref{eq:collision_step}, which can however be dealt with in special cases by means of suitable asymptotic procedures. One of the most effective is the so-called \textit{quasi-invariant interaction limit}, introduced in~\cite{cordier2005JSP} and reminiscent of the grazing collision limit applied in the classical kinetic theory~\cite{villani1998ARMA}. In our case, it consists in assuming that the binary rules~\eqref{eq:binary.controlled} produce a very small change of headway each time that two vehicles interact, so as to create a balance between the weakness of the interactions and the high frequency $1/\epsilon$ at which they occur.

To obtain such an effect, we scale the parameters $a$, $\nu$, $\sigma^2$ in~\eqref{eq:binary.controlled} using $\epsilon$ as follows:
\begin{equation}
	a=\frac{1}{\sqrt{\epsilon}}, \qquad \nu=\frac{1}{\epsilon}, \qquad \sigma^2=\epsilon.
	\label{eq:quasi-invariant_scaling}
\end{equation}
In this way, for $\epsilon$ small the uncontrolled part of the interaction becomes small as well, because $a$ is large and the variance of $\eta$ is small, cf.~\eqref{eq:binary}; at the same time, the control $\ucontr^\ast$ is dampened, because its penalisation $\nu$ is large, cf.~\eqref{eq:ustar}. On the whole, $s'\approx s$ and the interaction is quasi-invariant. In such a regime, we can expand
$$ \ave{\varphi(s')-\varphi(s)}=\varphi'(s)\ave{s'-s}+\frac{1}{2}\varphi''(s)\ave{(s'-s)^2}
	+\frac{1}{6}\varphi'''(\bar{s})\ave{(s'-s)^3}, $$
where $\bar{s}\in(\min\{s,\,s'\},\,\max\{s,\,s'\})$. Using~\eqref{eq:binary.controlled} for $s'$ with the scaling~\eqref{eq:quasi-invariant_scaling} gives:
\begin{align*}
	\ave{\varphi(s')-\varphi(s)} &= \varphi'(s)\epsilon\left\{\frac{1}{1+\epsilon\Theta^2}(s_\ast-s)
		+\frac{\Theta^2}{1+\epsilon\Theta^2}\left[\mu s_d(\rho)+(1-\mu)s_\ast-s\right]\right\} \\
	&\phantom{=} +\frac{\epsilon}{2}\varphi''(s)s^2+o(\epsilon),
\end{align*}
where $o(\epsilon)$ denotes, for every $s,\,s_\ast\in\R_+$, a remainder negligible with respect to $\epsilon$. Plugging this expansion into~\eqref{eq:collision_step} and considering that $\Theta^2\sim\operatorname{Bernoulli}(p)$, because so is $\Theta$ by assumption, we can fruitfully approximate~\eqref{eq:collision_step} with the limit equation arising for $\epsilon\to 0^+$:
\begin{align}
	\begin{aligned}[b]
		\partial_t & \int_{\R_+}\varphi(s)f(t,\,x,\,s)\,ds \\
		&= \int_{\R_+}\int_{\R_+}\varphi'(s)\{s_\ast-s+p\left[\mu s_d(\rho)+(1-\mu)s_\ast-s\right]\}f(t,\,x,\,s)f(t,\,x,\,s_\ast)\,ds\,ds_\ast \\
		&\phantom{=} +\frac{\rho}{2}\int_{\R_+}\varphi''(s)s^2f(t,\,x,\,s)\,ds,
	\end{aligned}
	\label{eq:FP.weak}
\end{align}
provided $f$ has bounded $s$-moments of order $3+\delta$ for some $\delta>0$ and $\ave{\abs{\eta}^3}<+\infty$, so as to guarantee the integrability of the terms contained in the aforementioned remainder $o(\epsilon)$.

Taking $\varphi(s)=1$ and recalling the definition~\eqref{eq:rho}, we see from~\eqref{eq:FP.weak} that the vehicle density $\rho$ is conserved in time by the interactions. Conversely, taking $\varphi(s)=s$ and recalling the definition~\eqref{eq:h}, we see from~\eqref{eq:FP.weak} that the mean headway is not if $p,\,\mu>0$, i.e. if the control~\eqref{eq:ustar} is truly active on at least some vehicles in the traffic stream and is not completely biased towards the speed alignment. Specifically, the mean headway $h$ satisfies the equation:
$$ \partial_t h=p\mu\rho\left(s_d(\rho)-h\right). $$
Omitting for simplicity the parametric dependence on $x$, we have explicitly
\begin{equation}
	h(t)=\left(h_0-s_d(\rho)\right)e^{-p\mu\rho t}+s_d(\rho),
	\label{eq:h.time}
\end{equation}
where $h_0:=h(0)$.

Integrating by parts the right-hand side of~\eqref{eq:FP.weak}, together with the following boundary conditions:
\begin{align}
	\begin{aligned}[c]
		& f(\cdot,\,\cdot,\,0)=\partial_sf(\cdot,\,\cdot,\,0)=0 \\
		& f(\cdot,\,\cdot,\,s),\,\partial_sf(\cdot,\,\cdot,\,s)\to 0 \text{ for } s\to +\infty \text{ (both infinitesimal of order $>2$)},
	\end{aligned}
	\label{eq:FP.bc}
\end{align}
we recognise that~\eqref{eq:FP.weak} is the weak form of the following \textit{Fokker-Planck equation} with time-varying coefficients:
\begin{equation}
	\partial_tf=\rho\partial_s\biggl\{\Bigl[(1+p)(s-h(t))-p\mu(s_d(\rho)-h(t))\Bigr]f\biggr\}+\frac{\rho}{2}\partial^2_s(s^2f).
	\label{eq:FP}
\end{equation}

\begin{figure}[!t]
\centering
\subfigure[]{\includegraphics[width=0.495\textwidth]{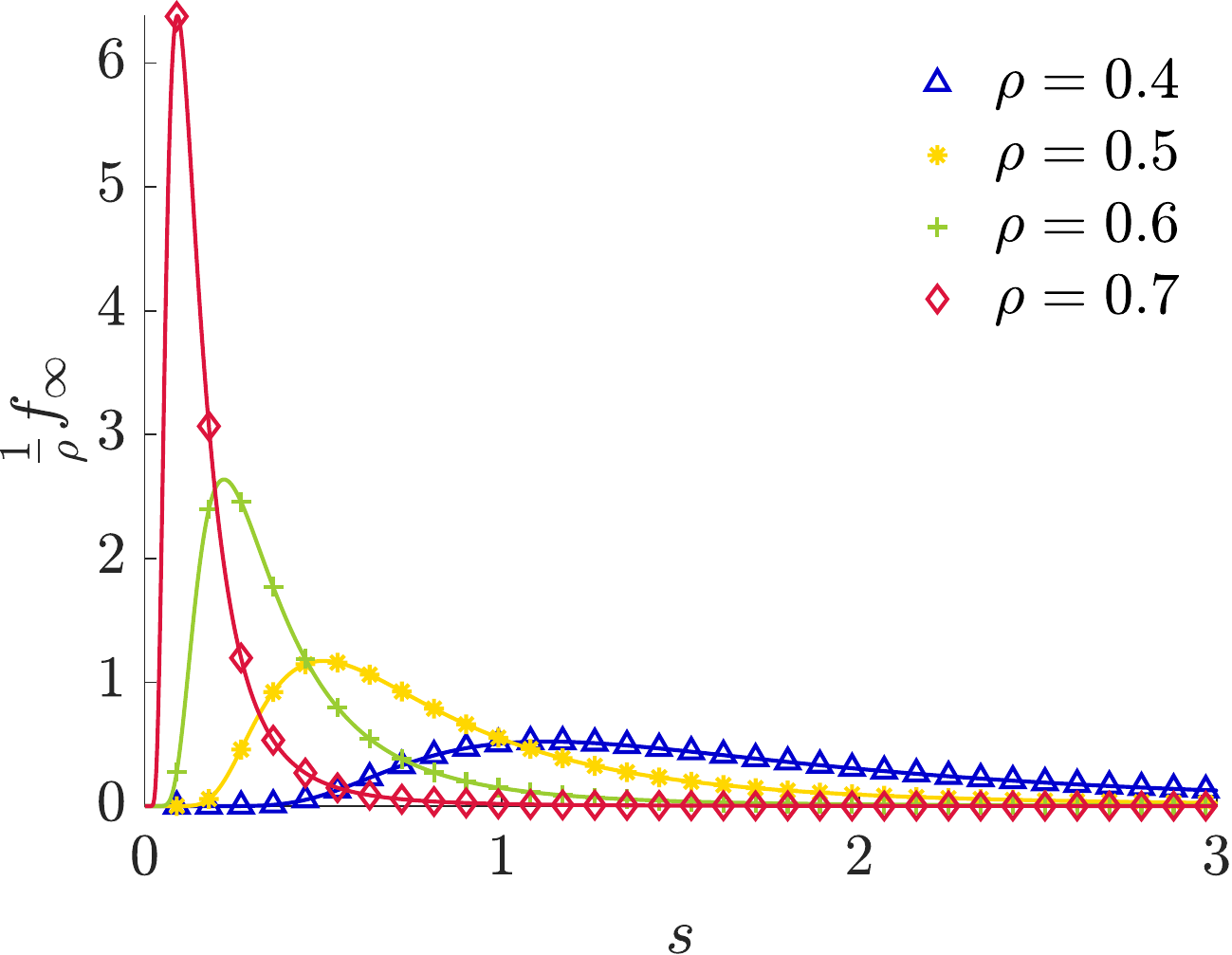}}
\subfigure[]{\includegraphics[width=0.495\textwidth]{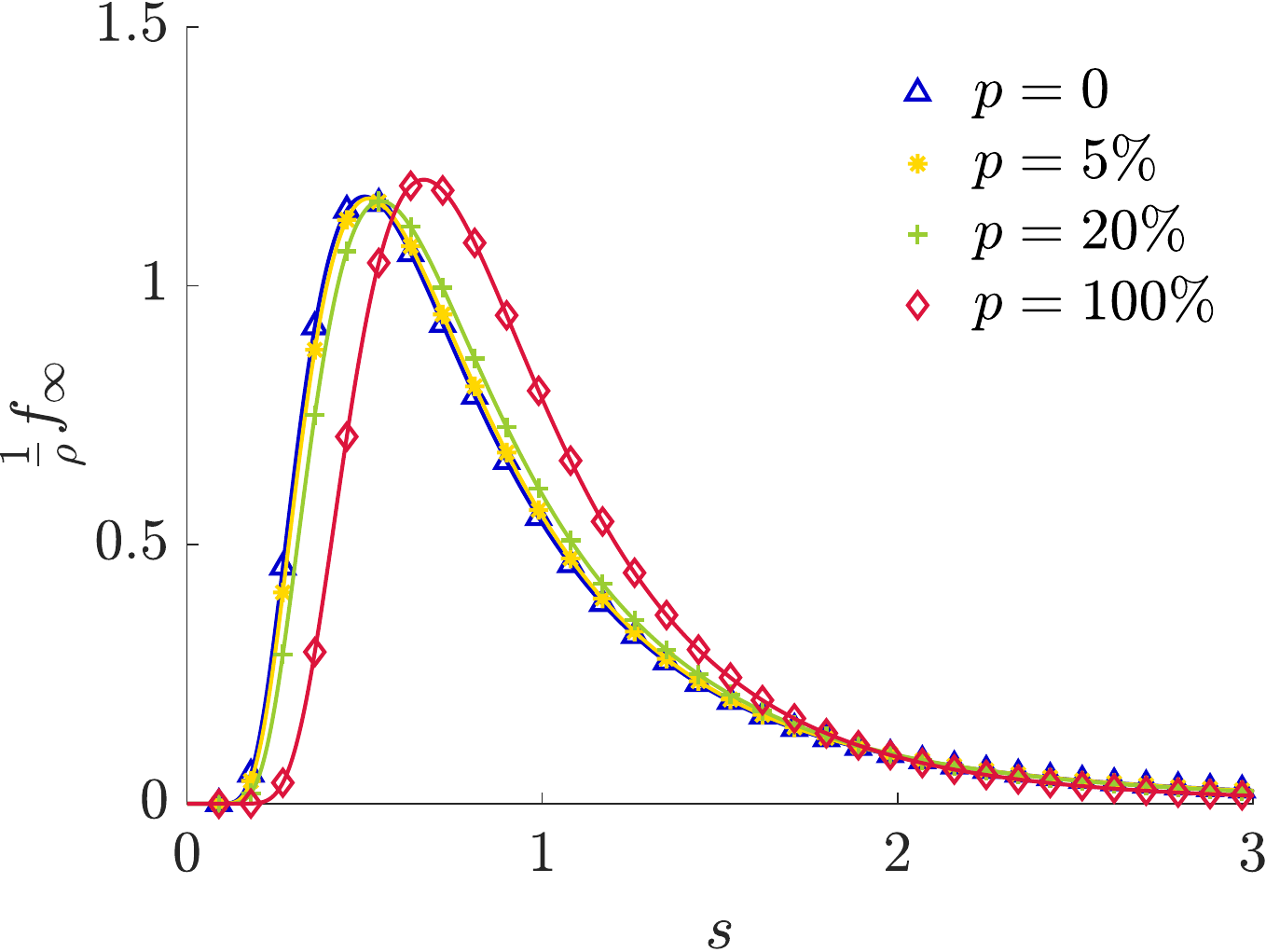}}
\subfigure[]{\includegraphics[width=0.495\textwidth]{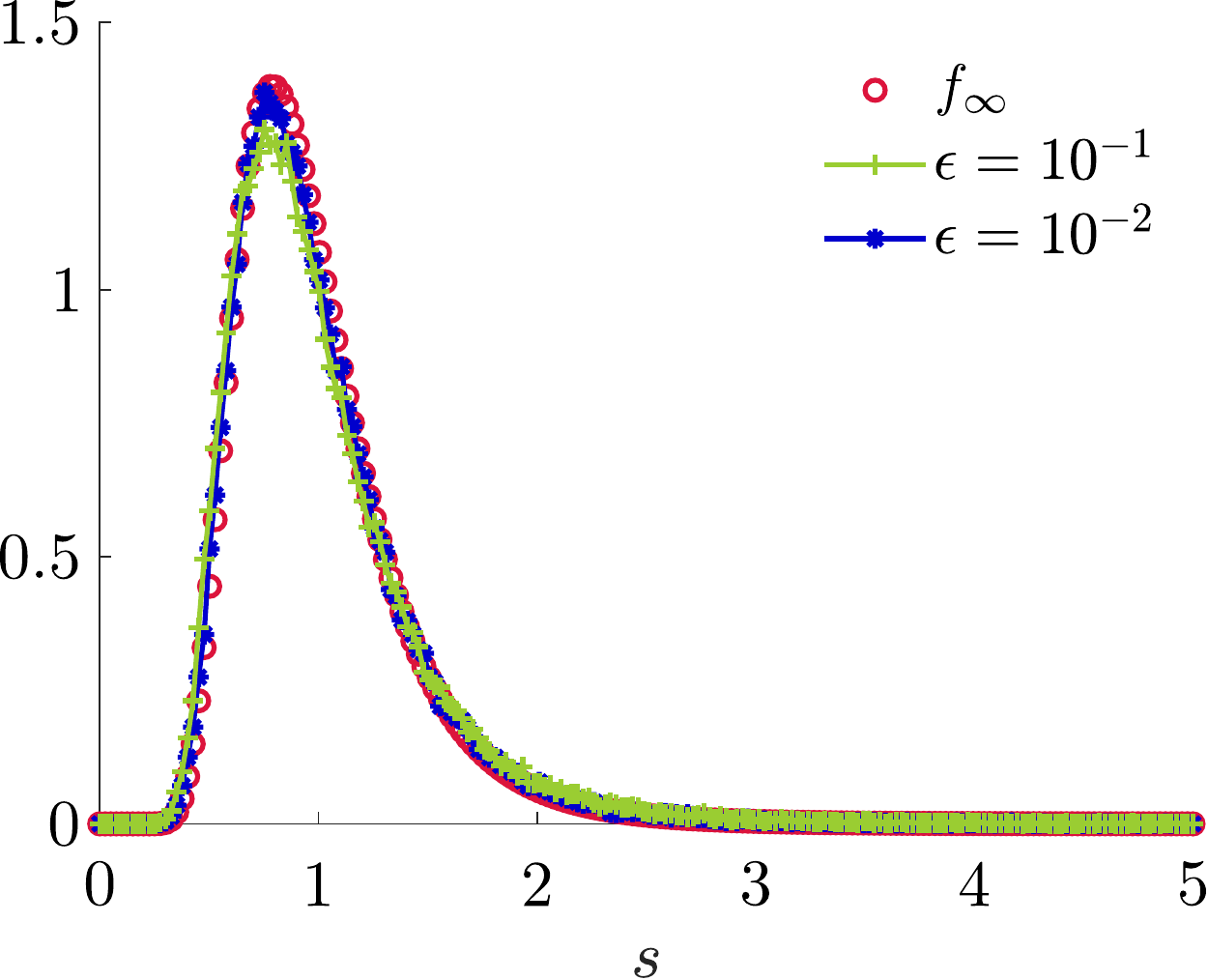}}
\caption{(a) The normalised distribution of the local asymptotic headway, cf.~\eqref{eq:finf}, plotted for $p=0$ (uncontrolled case), $s_d(\rho)=(1/\rho-1)^2$ and various local traffic densities. (b) The same normalised distribution plotted for $\rho=0.5$ and various penetration rates. (c) Comparison between $f^\infty$ given by~\eqref{eq:finf} and the stationary distribution to~\eqref{eq:collision_step} computed numerically for $\rho=0.5$ and for small values of $\epsilon$ in the quasi-invariant scaling~\eqref{eq:quasi-invariant_scaling}.}
\label{fig:finf}
\end{figure}

Under the quasi-invariant scaling~\eqref{eq:quasi-invariant_scaling}, for each fixed $t>0$ the solution to~\eqref{eq:FP} approximates the large time aggregate trend of the interactions in the point $x$. In particular, the local equilibrium distribution of the system is well approximated by the local asymptotic solution to~\eqref{eq:FP}, say $f^\infty(x,\,s)$, obtained formally for $t\to +\infty$. The advantage is that the latter can be computed explicitly, unlike the stationary solution to the original Boltzmann-type equation~\eqref{eq:collision_step}. In particular, observing from~\eqref{eq:h.time} that $h\to s_d(\rho)$ for $t\to +\infty$, $f^\infty$ solves the equation
$$ (1+p)\Bigl[(s-s_d(\rho))f^\infty\Bigr]+\frac{1}{2}\partial_s(s^2f^\infty)=0, $$
whence, still omitting the parametric dependence on $x$ and imposing $\int_{\R_+}f^\infty(s)\,ds=\rho$, we determine
\begin{equation}
	f^\infty(s)=\rho\frac{{(2(1+p)s_d(\rho))}^{3+2p}}{\Gamma(3+2p)}
		\cdot\frac{e^{-\frac{2(1+p)s_d(\rho)}{s}}}{s^{2(2+p)}},
	\label{eq:finf}
\end{equation}
where $\Gamma$ denotes the Gamma function. Notice that~\eqref{eq:finf} is an inverse Gamma distribution and that the only macroscopic quantity which parametrises it is the one conserved by the interactions, namely the density $\rho$, see Figure~\ref{fig:finf}(a, b). Figure~\ref{fig:finf}(c) confirms that, for $\epsilon$ sufficiently small, the stationary solution to the Boltzmann-type equation~\eqref{eq:collision_step} in the quasi-invariant scaling~\eqref{eq:quasi-invariant_scaling}, obtained numerically by means of a Monte Carlo scheme, is indeed very much well approximated by~\eqref{eq:finf}.

\begin{remark}
For $s\to +\infty$, the distribution function~\eqref{eq:finf} exhibits a \textit{Pareto-type} fat tail, similarly to the distribution curves of wealth and other social determinants studied in~\cite{cordier2005JSP,gualandi2018ECONOMICS}. Interestingly, such a fat-tailed trend has been found in some experimental data on the headway distribution in highway traffic, cf.~\cite{abuelenin2015IEEE}, and related to the presence of high occupancy vehicles in the traffic stream. Moreover, the standard deviation of the normalised distribution $\frac{1}{\rho}f^\infty(s)$, which, from the known formulas for the statistical moments of inverse Gamma random variables, is $s_d(\rho)/\sqrt{1+2p}$, depends linearly on the mean headway $s_d(\rho)$, a fact which is in turn reported as an experimental observation in~\cite{abuelenin2015IEEE}.
\end{remark}

\subsection{Analysis of the local traffic features}
\label{sect:local_traffic}
We now use the local equilibrium distribution~\eqref{eq:finf} to investigate in more detail the effect of the control mechanism on the local traffic.

\subsubsection{Speed distribution}
\begin{figure}[!t]
\centering
\subfigure[]{\includegraphics[width=0.495\textwidth]{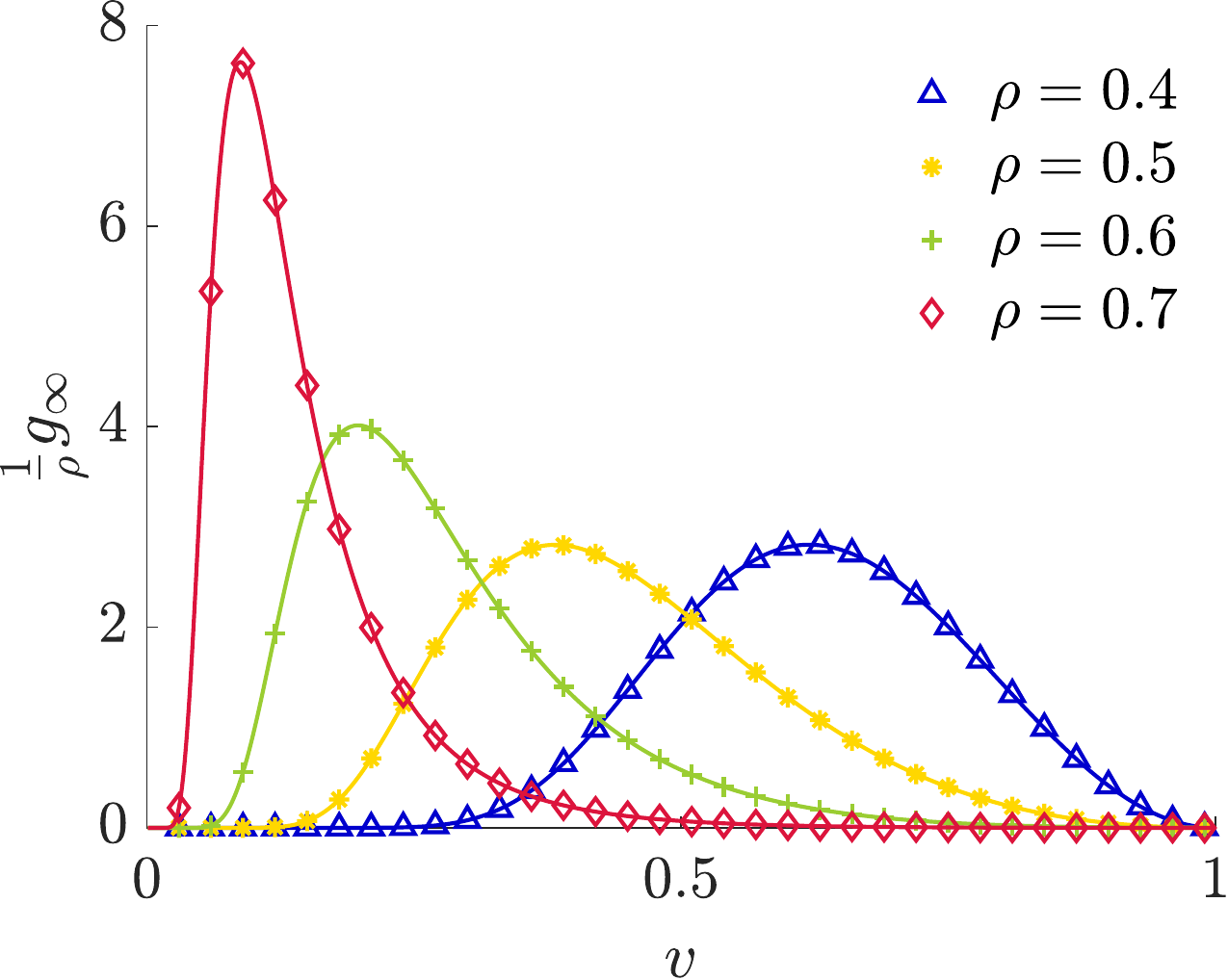}}
\subfigure[]{\includegraphics[width=0.495\textwidth]{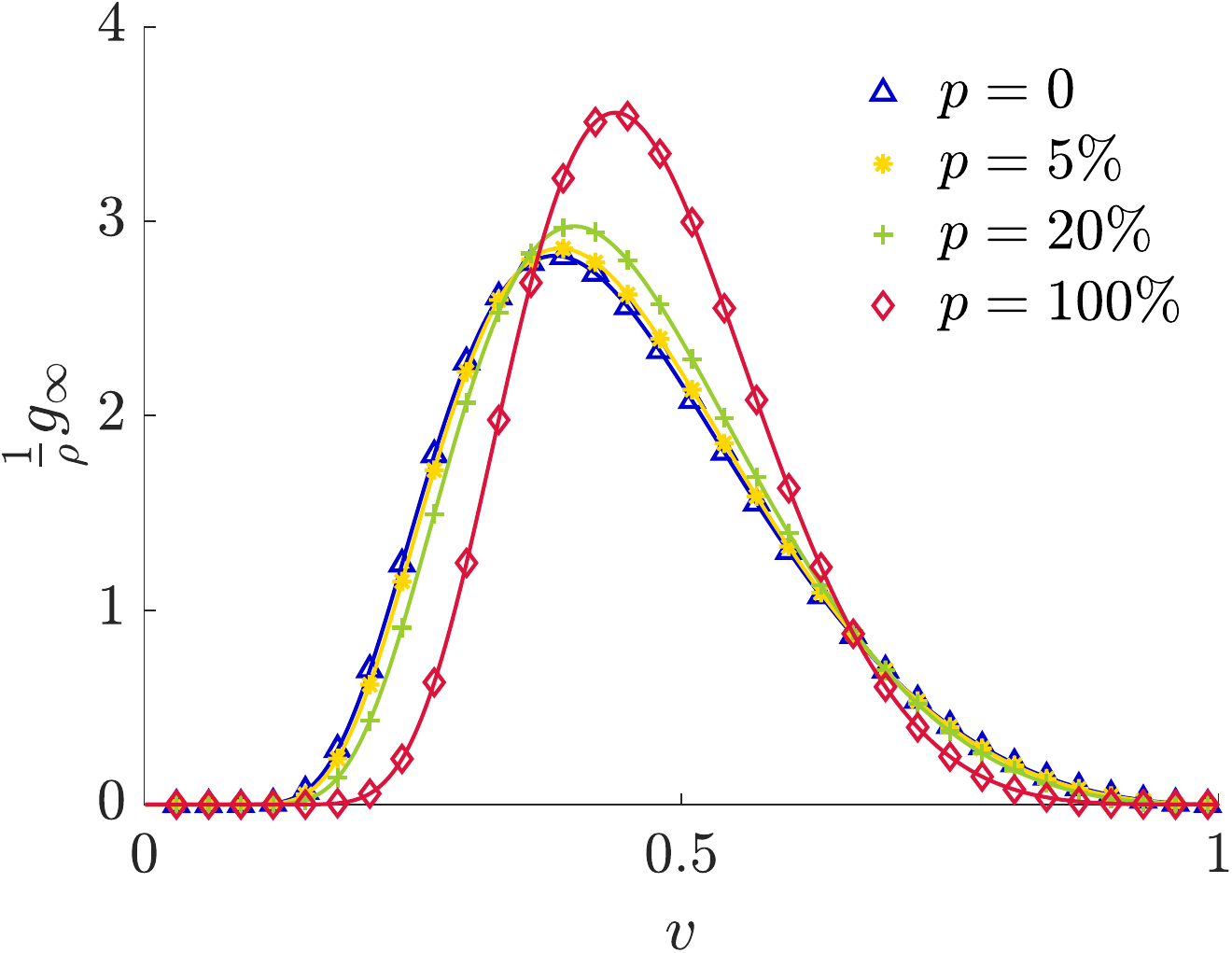}}
\caption{The normalised distribution of the local asymptotic speed, cf.~\eqref{eq:ginf}, plotted with the same parameters as in Figure~\ref{fig:finf} for: (a) various values of the local traffic density $\rho$ and $p=0$; (b) $\rho=0.5$ and various values of the penetration rate $p$.}
\label{fig:ginf}
\end{figure}

First of all, we notice that from~\eqref{eq:finf}, together with the relationship~\eqref{eq:v-s} between the headway $s$ and the speed $v$, we can obtain explicitly the local equilibrium distribution of the speed of the vehicles, say $g^\infty=g^\infty(v)$:
\begin{align}
	\begin{aligned}[b]
		g^\infty(v) &:= f^\infty\left(\frac{av}{1-v}\right)\cdot\frac{a}{{(1-v)}^2} \\
		&= \rho\frac{{\left(\frac{2(1+p)s_d(\rho)}{a}\right)}^{3+2p}}{\Gamma(3+2p)}
			\cdot\frac{1}{v^2}{\left(\frac{1-v}{v}\right)}^{2(1+p)}e^{-\frac{2(1+p)s_d(\rho)}{a}\cdot\frac{1-v}{v}},
	\end{aligned}
	\label{eq:ginf}
\end{align}
which is supported in the bounded interval $[0,\,1]$, see Figure~\ref{fig:ginf}. Notice that, consistently with the scaling~\eqref{eq:quasi-invariant_scaling}, the parameter $a$ brought into the expression of $g^\infty$ by the relationship~\eqref{eq:v-s} has to be thought of as sufficiently large.

Next, we observe that~\eqref{eq:finf},~\eqref{eq:ginf} are independent of the parameter $\mu$ of the binary interactions~\eqref{eq:binary.controlled}. This means that the simultaneous control of the safety distance and of the speed alignment, cf.~\eqref{eq:J}, is asymptotically equivalent to the control of the sole safety distance, or, in other words, that, in the long run, any $\mu>0$ is equivalent to $\mu=1$. Therefore, it is natural to wonder whether, for $\mu>0$, the sole control of the safety distance impacts automatically also on the speed variations in the traffic stream, which, by the way, have been recognised as responsible for increased levels of crash risk~\cite{peden2004WHO,WHO2015report}. We address this issue by means of the following result:

\begin{lemma} \label{lemma:VarV.VarS}
Let $S\in\R_+$ be the random variable expressing the equilibrium headway of any vehicle from its leading vehicle, with probability density function $\frac{1}{\rho}f^\infty$, cf.~\eqref{eq:finf}. Let moreover
$$ V:=\frac{S}{a+S}\in [0,\,1], $$
cf.~\eqref{eq:v-s}, be the random variable giving the speed of that vehicle. Then:
$$ \Var(V)\leq\frac{1}{a+s_d(\rho)}\left(\frac{2s_d(\rho)}{a}+\frac{a^2}{{(a+s_d(\rho))}^3}\sqrt{\Var(S)}\right)\sqrt{\Var(S)}, $$
where $\Var(S)$, $\Var(V)$ denote the variance of $S$, $V$, respectively.
\end{lemma}
\begin{proof}
Since for $S\geq 0$ the function $S\mapsto\frac{S}{a+S}$ is concave, its graph lies below any of its tangent lines, cf. Figure~\ref{fig:v-s}. Hence, considering the tangent at $S=s_d(\rho)$, we have
\begin{equation}
	V\leq\frac{1}{a+s_d(\rho)}\left(s_d(\rho)+\frac{a}{a+s_d(\rho)}(S-s_d(\rho))\right),
	\label{eq:V.leq.S}
\end{equation}
which, since $\E(S)=s_d(\rho)$, implies
\begin{align*}
	\Var(V) &= \E(V^2)-\E^2(V) \\
	&\leq {\left(\frac{s_d(\rho)}{a+s_d(\rho)}\right)}^2+\frac{a^2}{(a+s_d(\rho))^4}\E((S-s_d(\rho))^2)
		+2\frac{as_d(\rho)}{(a+s_d(\rho))^3}\E(S-s_d(\rho))-\E^2(V) \\
	&=\frac{a^2}{{(a+s_d(\rho))}^4}\Var(S)+{\left(\frac{s_d(\rho)}{a+s_d(\rho)}\right)}^2-\E^2(V)
\intertext{and further, considering that from~\eqref{eq:V.leq.S} it results $\E(V)\leq\frac{s_d(\rho)}{a+s_d(\rho)}$,}
	&\leq \frac{a^2}{(a+s_d(\rho))^4}\Var(S)+\frac{2s_d(\rho)}{a+s_d(\rho)}\left(\frac{s_d(\rho)}{a+s_d(\rho)}-\E(V)\right).
\end{align*}

Let us examine more closely the term in parenthesis on the right-hand side. We have:
\begin{align*}
	\frac{s_d(\rho)}{a+s_d(\rho)}-\E(V) &= \abs{\frac{s_d(\rho)}{a+s_d(\rho)}-\E(V)} \\
	&\leq \frac{1}{\rho}\int_{\R_+}\abs{\frac{s_d(\rho)}{a+s_d(\rho)}-\frac{s}{a+s}}f^\infty(s)\,ds \\
	&\leq \frac{1}{a\rho}\int_{\R_+}\abs{s_d(\rho)-s}f^\infty(s)\,ds \\
	&\leq \frac{1}{a}{\left(\frac{1}{\rho}\int_{\R_+}\abs{s_d(\rho)-s}^2f^\infty(s)\,ds\right)}^{\frac{1}{2}}
		=\frac{1}{a}\sqrt{\Var(S)},
\end{align*}
where in the last passage we have used the Cauchy-Schwartz inequality with respect to the probability measure $\frac{1}{\rho}f^\infty(s)ds$. Using this in the previous calculation, we finally get the thesis.
\end{proof}

\begin{figure}[!t]
\centering
\subfigure[]{\includegraphics[width=0.495\textwidth]{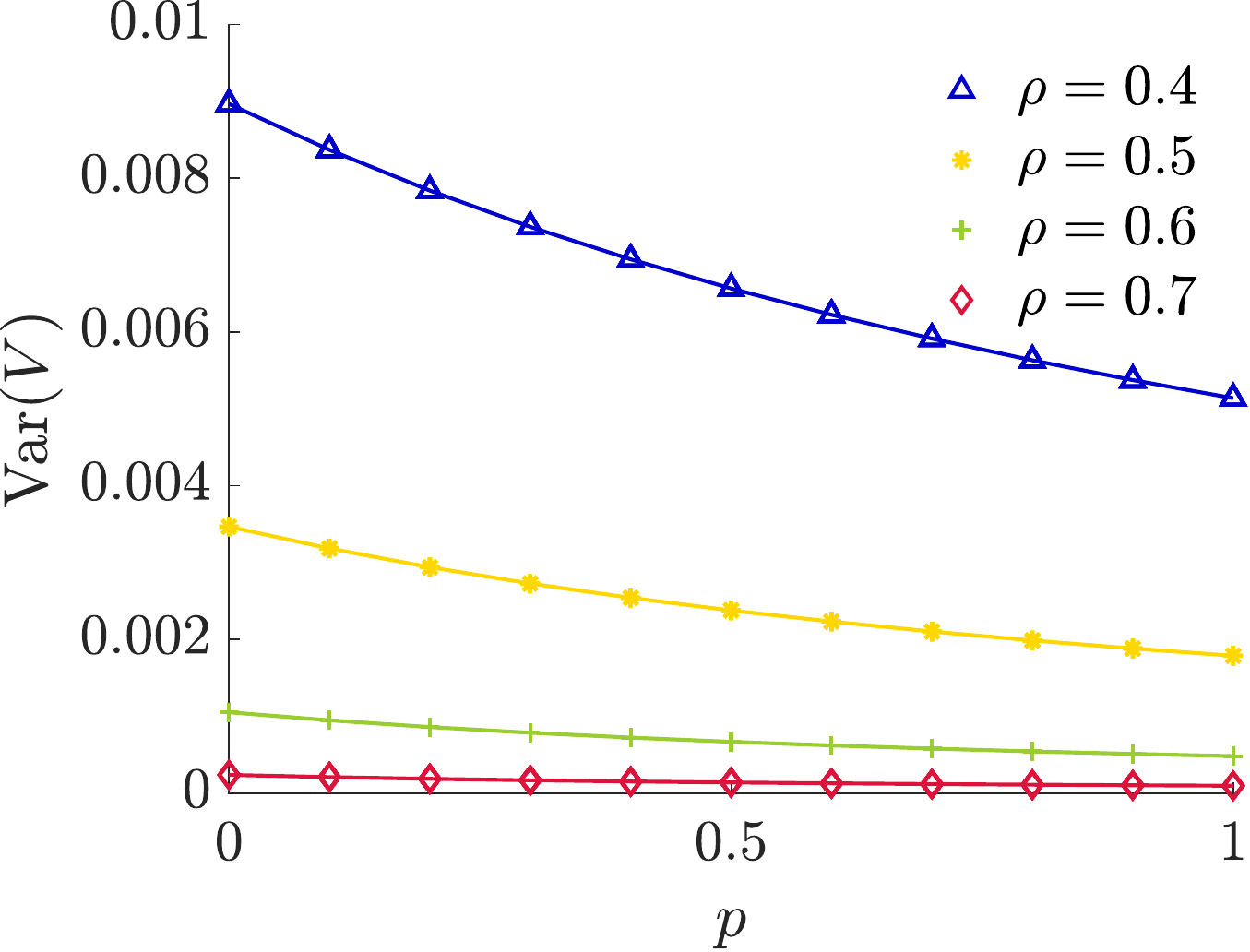}}
\subfigure[]{\includegraphics[width=0.495\textwidth]{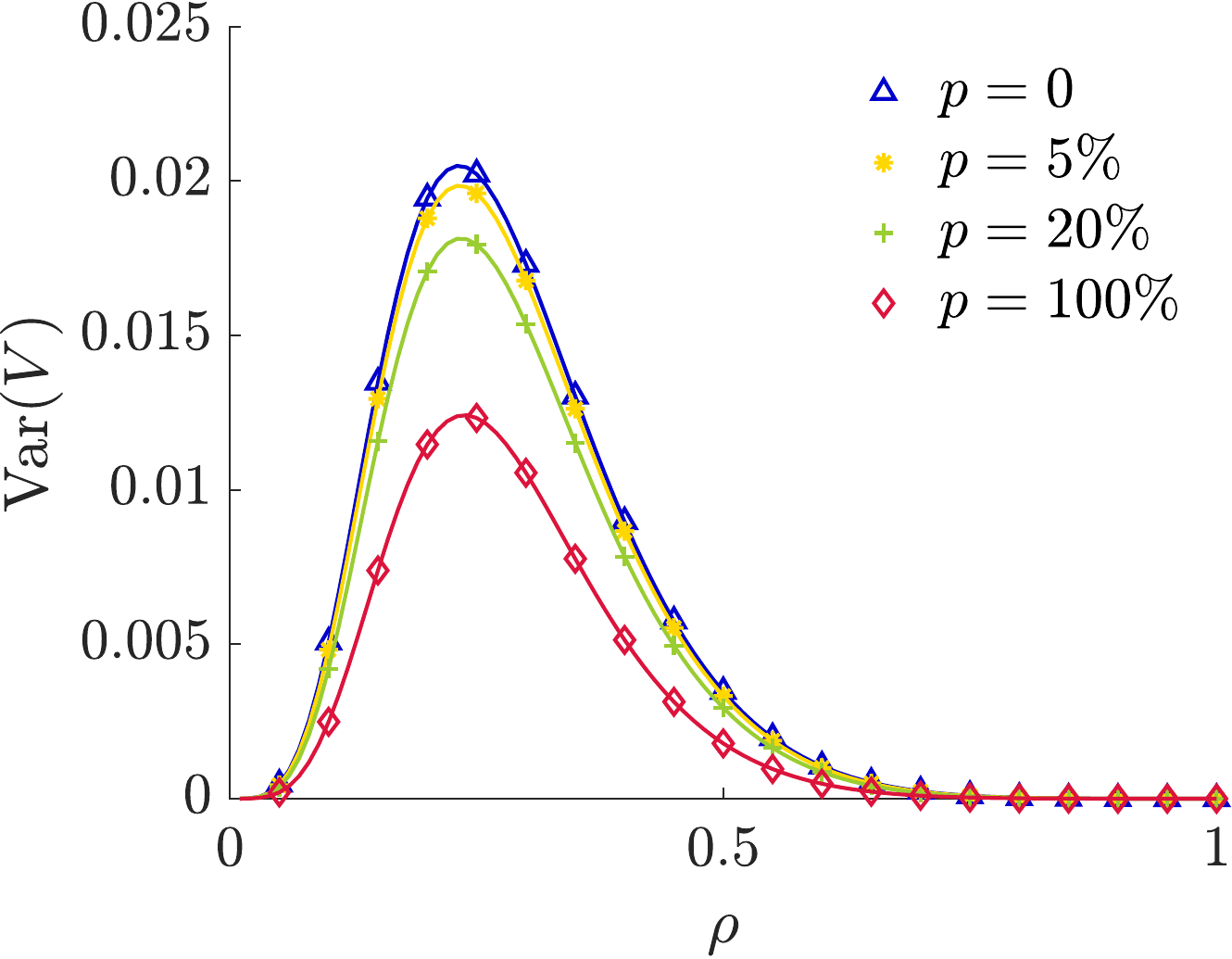}}
\caption{The local asymptotic speed variance as a function of: (a) the penetration rate $p$ for various traffic densities $\rho$; (b) the traffic density $\rho$ for various penetration rates. The curves have been computed numerically from~\eqref{eq:ginf} taking $s_d(\rho)=\left(1/\rho-1\right)^2$ and $a=10$, which corresponds to $\epsilon=10^{-2}$ in~\eqref{eq:quasi-invariant_scaling}.}
\label{fig:varV}
\end{figure}

Owing to Lemma~\ref{lemma:VarV.VarS}, we have\footnote{We use the notation $a\lesssim b$ to mean that there exists a constant $C>0$, whose specific value is unimportant, such that $a\leq Cb$.}
\begin{equation}
	\Var(V)\lesssim\left(1+\sqrt{\Var(S)}\right)\sqrt{\Var(S)}
		=\left(1+\frac{s_d(\rho)}{\sqrt{1+2p}}\right)\frac{s_d(\rho)}{\sqrt{1+2p}},
	\label{eq:varV}
\end{equation}
whence we conclude that increasing $p>0$ leads eventually to a reduction of the speed variance as confirmed by Figure~\ref{fig:varV}(a). Figure~\ref{fig:varV}(b) shows instead that the trend of the speed variance with respect to the traffic density $\rho$ may be, in general, more complex, in particular non-monotone, because it depends on the choice of the function $s_d(\rho)$, cf.~\eqref{eq:varV}. However, increasing the penetration rate $p$ leads systematically to a reduction of the speed dispersion.

\begin{figure}[!t]
\centering
\includegraphics[width=0.495\textwidth]{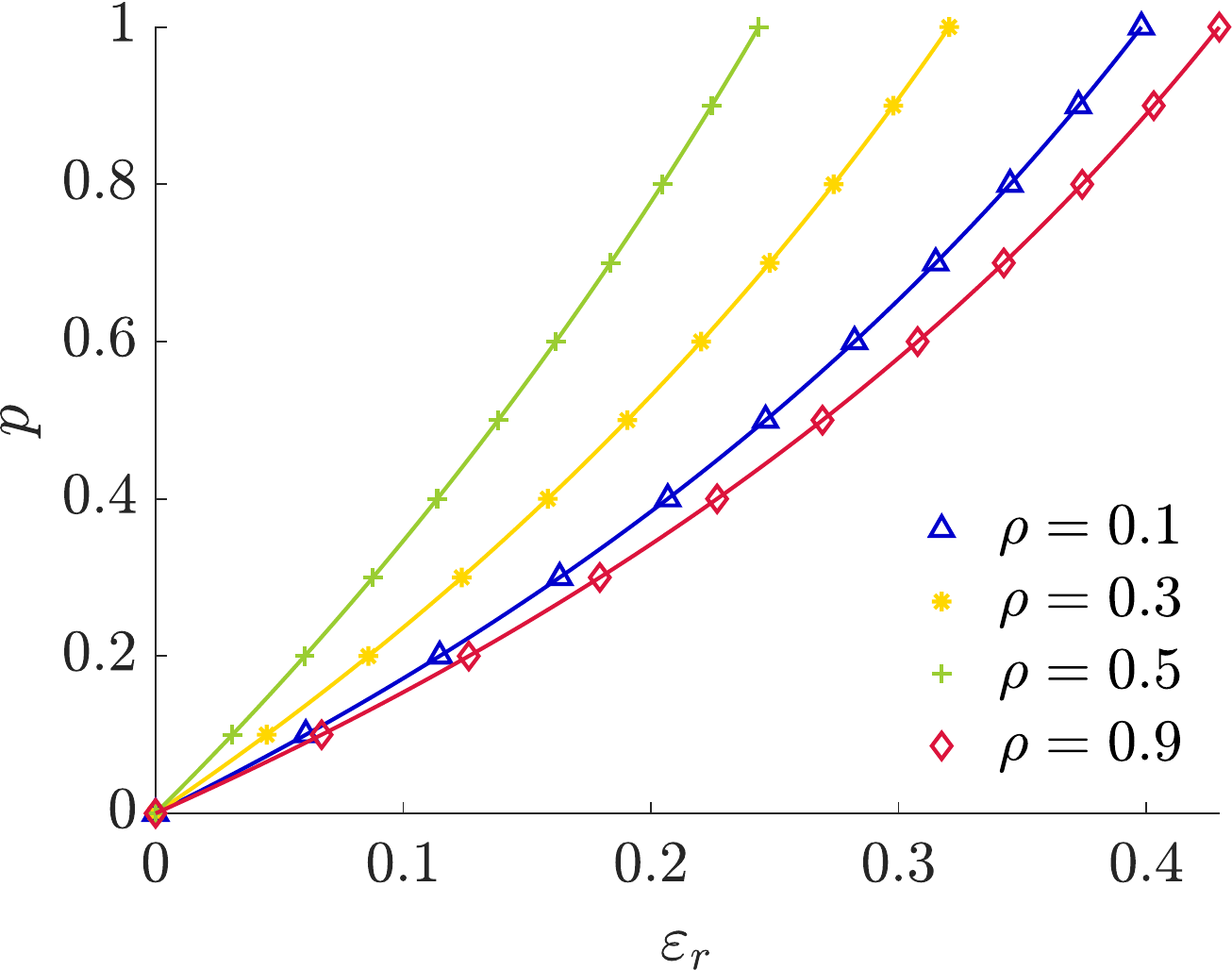}
\caption{The relationship between the relative speed variance reduction $\varepsilon_r$, cf.~\eqref{eq:epsr}, and the penetration rate $p$ necessary to achieve it at various traffic densities $\rho$. The parameters of the model are set like in Figure~\ref{fig:varV}.}
\label{fig:varV_rel}
\end{figure}

A quite natural question is what penetration rate is necessary in order to achieve a certain speed variance reduction with respect to the uncontrolled case. Denoting by $V_p$ the local equilibrium speed at the penetration rate $p$, this amounts to considering the relationship between $p$ and the \textit{relative speed variance reduction}
\begin{equation}
	\varepsilon_r:=1-\frac{\Var(V_p)}{\Var(V_0)},
	\label{eq:epsr}
\end{equation}
which is illustrated in Figure~\ref{fig:varV_rel}. Obviously, increasing $\varepsilon_r$ requires higher and higher penetration rates at all traffic densities. Nevertheless, since, by definition, $p$ cannot exceed $1$, we observe a saturation of $\varepsilon_r$, which cannot attain arbitrarily large values.

\bigskip

Let us now consider the case $\mu=0$, which corresponds to no control on the safety distance, cf.~\eqref{eq:J}. In this case, the interactions~\eqref{eq:binary.controlled} conserve locally the mean headway, as it can be easily checked from~\eqref{eq:FP.weak} with $\varphi(s)=s$. Consequently, the local equilibrium distribution reads
$$ f^\infty(s)=\rho\frac{{(2(1+p)h)}^{3+2p}}{\Gamma(3+2p)}
	\cdot\frac{e^{-\frac{2(1+p)h}{s}}}{s^{2(2+p)}}, $$
where $h\geq 0$ is the constant (in time) and locally arbitrary mean headway. This is also the local equilibrium distribution obtained for any $\mu\geq 0$ with $p=0$ (no control on any vehicle, i.e. $\Theta=0$ deterministically), for in the latter case the interactions~\eqref{eq:binary} conserve in turn the local mean headway. The result of Lemma~\ref{lemma:VarV.VarS} still holds with $s_d(\rho)$ replaced by $h$, whence
$$ \Var(V)\lesssim\left(1+\frac{h}{\sqrt{1+2p}}\right)\frac{h}{\sqrt{1+2p}}. $$
Assuming $\mu=0$, this formula suggests that for $p>0$ the speed dispersion is actually reduced with respect to the uncontrolled case $p=0$. However, it is not possible to compare in general the case $\mu=0$, $p>0$ with the case $\mu>0$, $p>0$ discussed before, because the local mean headways $h$, $s_d(\rho)$ are in principle uncorrelated.

\subsubsection{Time headway distribution}
Combining~\eqref{eq:finf} and~\eqref{eq:tau}, we can compute the local equilibrium distribution of the time headway (cf. Remark~\ref{rem:time_headway}), say $k^\infty=k^\infty(\tau)$:
\begin{align*}
	k^\infty(\tau) &:= f^\infty(\tau-a)H(\tau-a) \\
	&= \rho\frac{{(2(1+p)s_d(\rho))}^{3+2p}}{\Gamma(3+2p)}
		\cdot\frac{e^{-\frac{2(1+p)s_d(\rho)}{\tau-a}}}{(\tau-a)^{2(2+p)}}H(\tau-a),
\end{align*}
where $H$ is the unit step (Heaviside) function. The cutoff at $\tau=a>0$ of the support of $k^\infty$ (cf. also Remark~\ref{rem:time_headway}) is particularly in agreement with the experimental findings reported in~\cite{neubert1999PRE}. Again, for consistency with the scaling~\eqref{eq:quasi-invariant_scaling}, here the parameter $a$ has to be large enough.

\subsubsection{The role of~\texorpdfstring{$\boldsymbol{\mu}$}{}}
The previous analysis has revealed that the parameter $\mu$ does not affect the equilibrium distribution function $f^\infty$, hence it will not impact on the closure of the hydrodynamic equations. On the other hand, from~\eqref{eq:h.time} it is clear that $\mu$ determines the rate of convergence in time of the mean headway $h$ to its local equilibrium value $s_d(\rho)$. More in general, we claim that $\mu$ affects the rate of convergence of the whole distribution function $f$ to $f^\infty$. As such, it plays a role in the assumption underlying the local equilibrium closure, namely that the interactions reach quickly a local equilibrium.

A standard way of studying the convergence to equilibrium of the solutions to Fokker-Planck equations is to resort to \textit{entropy} functionals, which are readily available if the equation can be written in conservative form with a gradient-type flux, see e.g.~\cite{barbaro2016MMS}. Unfortunately, this is not our case, because the diffusion coefficient in~\eqref{eq:FP} is non-constant. For non-constant diffusion coefficients, an alternative approach is provided by \textit{relative entropy} functionals of the form
\begin{equation}
	\cH[f,\,f^\infty](t):=\int_{\R_+}\Phi(F(t,\,s))f^\infty(s)\,ds,
	\label{eq:H}
\end{equation}
cf.~\cite{furioli2017M3AS}, where $F:=f/f^\infty$ and $\Phi:\R_+\to\R$ is any smooth and strictly convex function. If $\Phi(1)=0$ then, owing to Jensen's inequality, it results $\cH\geq 0$, the equality holding only for $f=f^\infty$. If, furthermore, $\cH$ decreases along the solutions to~\eqref{eq:FP} then it is a Lyapunov functional for~\eqref{eq:FP} relative to the equilibrium $f^\infty$. The rate of decrease of $\cH$ estimates the rate of convergence of the solutions of~\eqref{eq:FP} to $f^\infty$.

Following the computations in~\cite{furioli2017M3AS} and using~\eqref{eq:h.time},~\eqref{eq:FP.bc} we determine:
\begin{equation}
	\frac{d}{dt}\cH[f,\,f^\infty](t)=-\frac{\rho}{2}\int_{\R_+}s^2\Phi''(F(t,\,s)){\left(\partial_sF(t,\,s)\right)}^2f^\infty(s)\,ds+\cI[f,\,f^\infty](t).
	\label{eq:H.time_derivative}
\end{equation}
The first term on the right-hand side is analogous to the one found in~\cite{furioli2017M3AS}. Owing to the convexity of $\Phi$, it is non-positive, thus it confers a non-increasing trend on the time derivative of $\cH$ along the solutions to~\eqref{eq:FP}. Conversely, the second term is computed as:
$$ \cI[f,\,f^\infty](t):=-\frac{\rho(1+p)}{2}\left(s_d(\rho)-h_0\right)\left(1-\frac{p\mu}{1+p}\right)e^{-p\mu\rho t}
	\int_{\R_+}\Phi''(F(t,\,s))\partial_sF^2(t,\,s)f^\infty(s)\,ds $$
and is clearly produced by the non-conservation of the mean headway (namely, by the fact that possibly $h_0\ne s_d(\rho)$). It does not have a counterpart in the results reported in~\cite{furioli2017M3AS}, because there only mean-preserving Fokker-Planck equations are considered. We observe that this term is indefinite, as the sign of $\partial_sF^2$ cannot be decided \textit{a priori}. Hence, we cannot guarantee, in general, that~\eqref{eq:H} is a Lyapunov functional for~\eqref{eq:FP} relative to $f^\infty$. Nevertheless, we also observe that, owing to the smoothness of $\Phi$ and to the boundary conditions~\eqref{eq:FP.bc}, we have
$$ \int_{\R_+}\Phi''(F(t,\,s))\partial_sF^2(t,\,s)f^\infty(s)\,ds=-2\int_{\R_+}\Phi'(F(t,\,s))\partial_sf(t,\,s)\,ds, $$
hence, if $F$ is bounded,
$$ \abs{\int_{\R_+}\Phi''(F(t,\,s))\partial_sF^2(t,\,s)f^\infty(s)\,ds}\lesssim\int_{\R_+}\abs{\partial_sf(t,\,s)}\,ds<+\infty, $$
because, owing to~\eqref{eq:FP.bc}, it results $\partial_sf(t,\,\cdot)\in L^1(\R_+)$. If furthermore $\partial_sf\in L^\infty(\R_+;\,L^1(\R_+))$ then\footnote{Notice that $1-\frac{p\mu}{1+p}>0$, because $0\leq p,\,\mu\leq 1$.}
$$ \abs{\cI[f,\,f^\infty](t)}\lesssim\left(1-\frac{p\mu}{1+p}\right)e^{-p\mu\rho t}, $$
therefore, as soon as $\mu>0$, $\cI$ vanishes exponentially fast for $t\to +\infty$ with a $\mu$-dependent rate, in such a way that the larger $\mu$ the smaller $\abs{\cI[f,\,f^\infty]}$ at all times. Hence, in the worst case, i.e. when at certain times $\cI$ is positive, a large $\mu$ increases the chances that $\cI$ be quickly dominated by the first term on the right-hand side of~\eqref{eq:H.time_derivative}, so that, after a possible transient, the interactions start to approach rapidly the local equilibrium $f^\infty$.

\begin{figure}[!t]
\centering
\subfigure[$\mu=0.1$]{\includegraphics[width=0.495\textwidth]{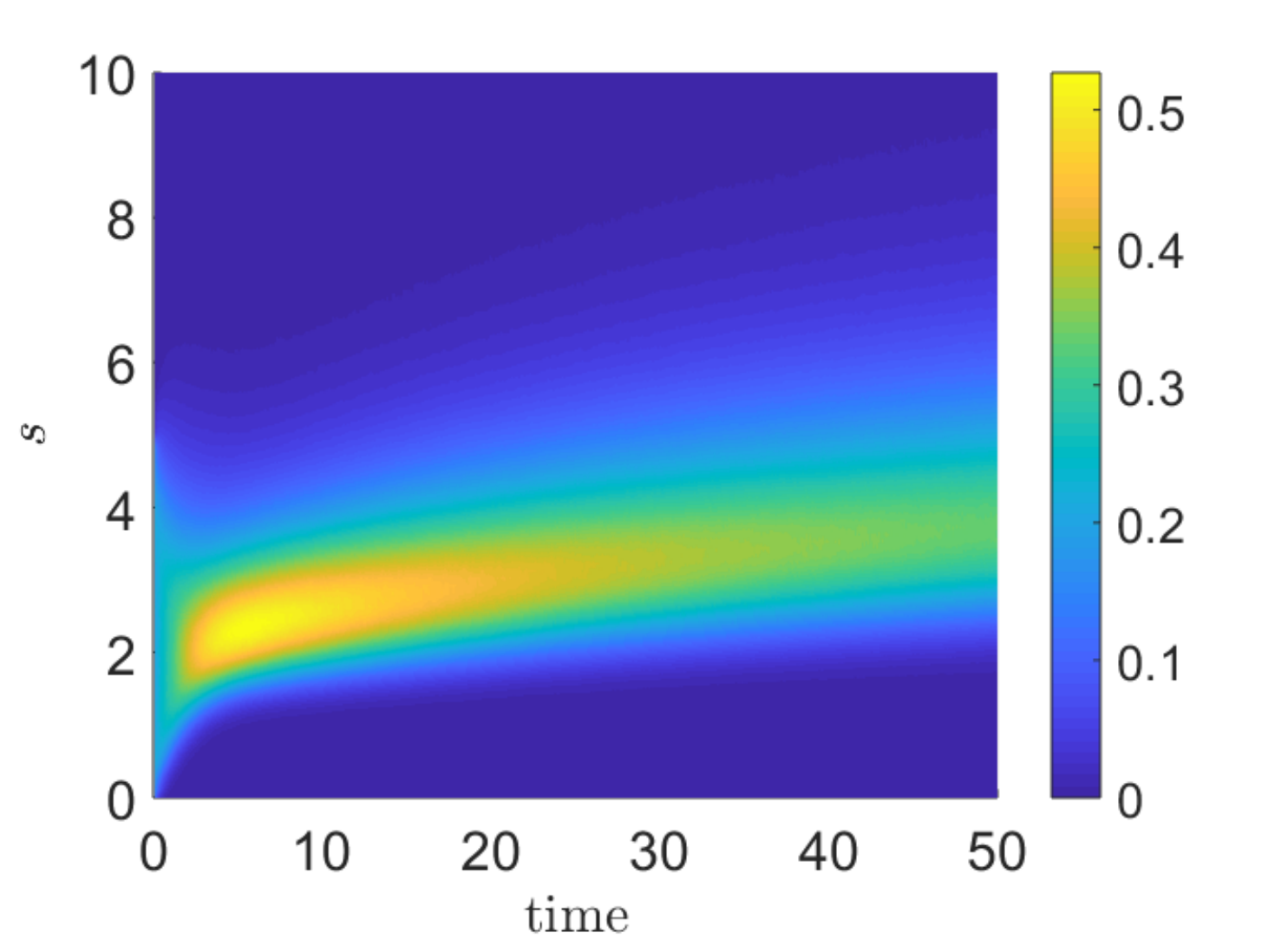}}
\subfigure[$\mu=1$]{\includegraphics[width=0.495\textwidth]{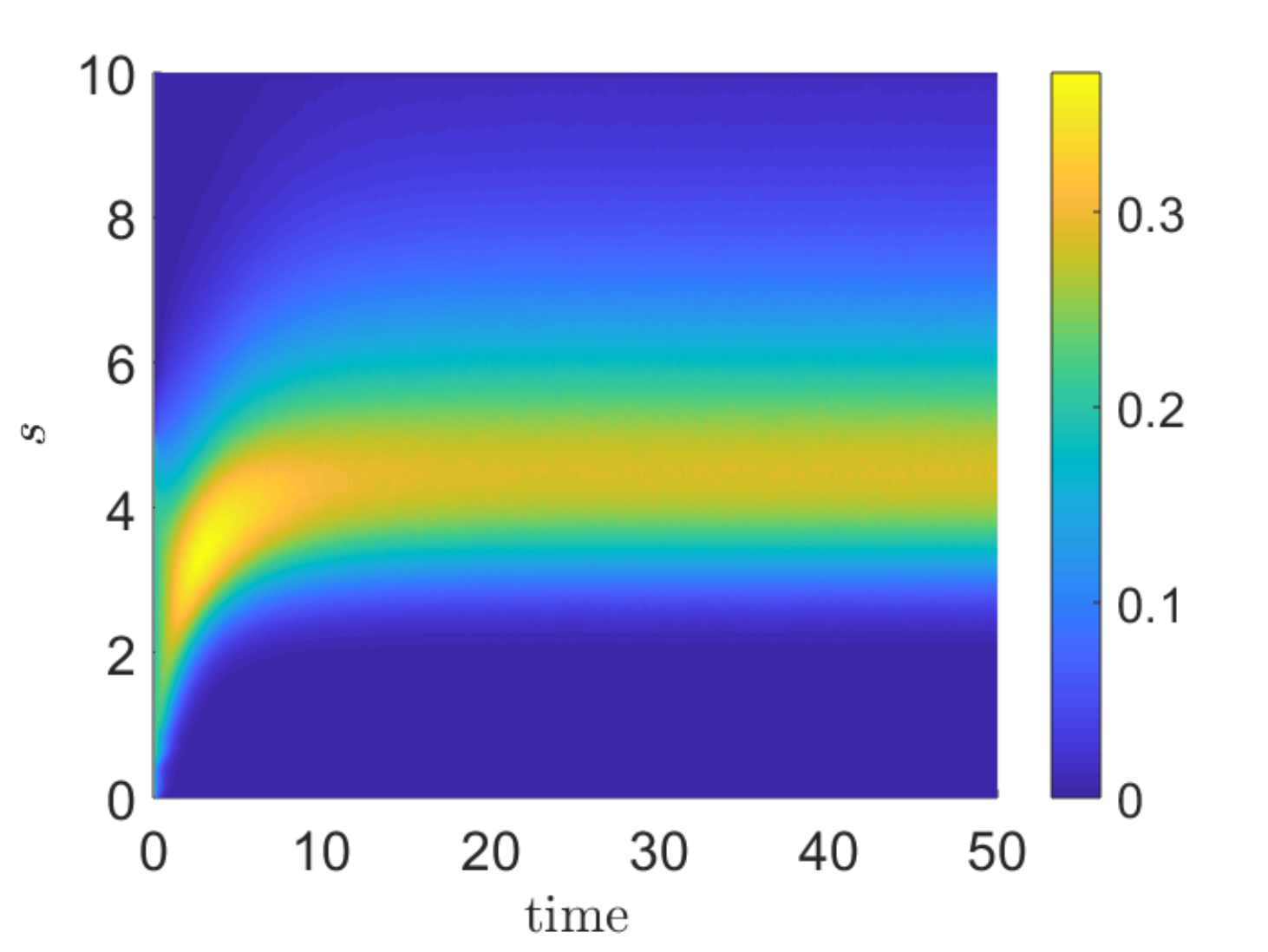}}
\caption{Contours of the local headway distribution $f(\cdot,\,x,\,\cdot)$ computed numerically from~\eqref{eq:collision_step} in the case $\rho=0.3$, $p=50\%$ and with: (a) $\mu=0.1$, (b) $\mu=1$.}
\label{fig:contours.f}
\end{figure}

Figure~\ref{fig:contours.f} provides evidence of the effect of $\mu$ on the rate of convergence of the local headway distribution $f$ to its equilibrium $f^\infty$. For $\mu=0.1$, the asymptotic distribution is not yet fully reached at $t=50$, cf. Figure~\ref{fig:contours.f}(a). Conversely, for $\mu=1$, the asymptotic distribution is reached at $t\approx 10$, cf. Figure~\ref{fig:contours.f}(b).

\subsection{Hydrodynamic limits}
\label{sect:hydro_lim}
We now solve the transport step~\eqref{eq:transport_step} of the splitting procedure, taking advantage of the local equilibrium distribution~\eqref{eq:finf}. We stress that, consistently with the Fokker-Planck asymptotic procedure which produced the Maxwellian~\eqref{eq:finf}, throughout this section the parameter $a$ appearing in~\eqref{eq:transport_step} has to be formally though of large enough, cf. also~\eqref{eq:quasi-invariant_scaling}.

For every time $t>0$ and at every point $x\in\R$, we let
\begin{equation}
	f(t,\,x,\,s)=\rho(t,\,x)\frac{{[2(1+p)s_d(\rho(t,\,x))]}^{3+2p}}{\Gamma(3+2p)}
		\cdot\frac{e^{-\frac{2(1+p)s_d(\rho(t,\,x))}{s}}}{s^{2(2+p)}}
	\label{eq:f.first_order}
\end{equation}
in~\eqref{eq:transport_step}, so that, choosing $\varphi(s)=1$, we get the conservation law
\begin{equation}
	\partial_t\rho+\partial_xq(\rho)=0
	\label{eq:first_order}
\end{equation}
with flux
\begin{equation}
	q(\rho):=\rho\frac{{\left(2(1+p)s_d(\rho)\right)}^{3+2p}}{\Gamma(3+2p)}
		\int_{\R_+}\frac{e^{-\frac{2(1+p)s_d(\rho)}{s}}}{(a+s)s^{3+2p}}\,ds.
	\label{eq:q.s}
\end{equation}
Since, as already observed, the traffic density $\rho$ is the only quantity conserved by the microscopic interactions~\eqref{eq:binary.controlled}, a single macroscopic equation is sufficient to have a self-consistent hydrodynamic model. We refer therefore to~\eqref{eq:first_order} as a \textit{first order} model.

With the substitution $v:=\frac{s}{a+s}$, and recalling also the relationship~\eqref{eq:ginf}, the flux can be given the equivalent alternative form
\begin{equation}
	q(\rho)=\rho\frac{{\left(\frac{2(1+p)s_d(\rho)}{a}\right)}^{3+2p}}{\Gamma(3+2p)}\int_0^1\frac{1}{v}\left(\frac{1-v}{v}\right)^{2(1+p)}
		e^{-\frac{2(1+p)s_d(\rho)}{a}\cdot\frac{1-v}{v}}\,dv,
	\label{eq:q.v}
\end{equation}
which is more convenient for numerical approximation purposes, because, unlike~\eqref{eq:q.s}, it involves an integral on a bounded interval.

\begin{figure}[!t]
\centering
\includegraphics[width=0.495\textwidth]{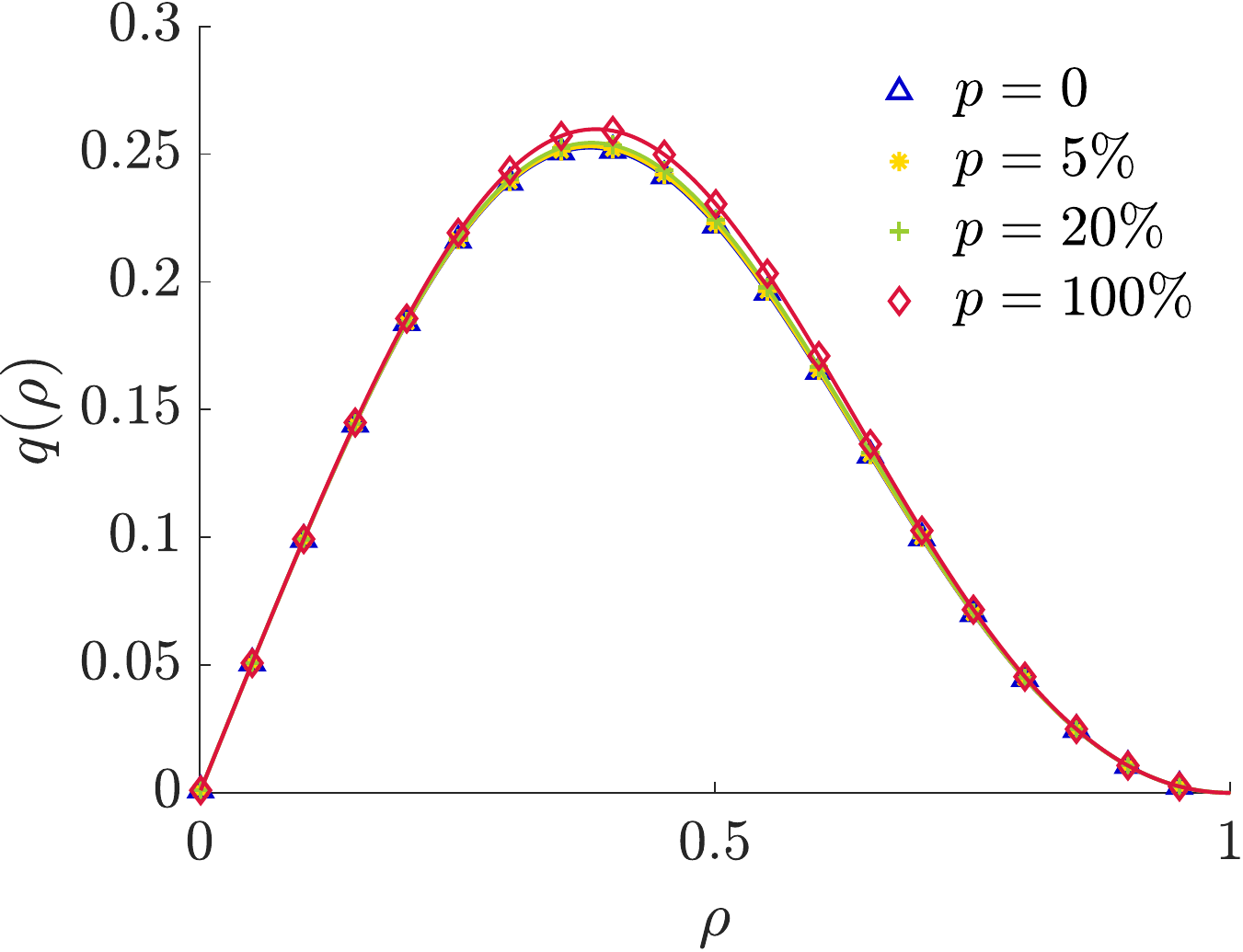}
\caption{The fundamental diagram $\rho\mapsto q(\rho)$, cf.~\eqref{eq:q.v}, for various penetration rates $p$. The parameters of the model are set like in Figure~\ref{fig:varV}.}
\label{fig:funddiag}
\end{figure}

The mapping $\rho\mapsto q(\rho)$ defines the \textit{fundamental diagram} of traffic, which we plot in Figure~\ref{fig:funddiag} for different values of the penetration rate $p$. Interestingly, and consistently with the experimental findings reported in~\cite{stern2018TRC}, the macroscopic flux is virtually unaffected by the driver-assist vehicles at all penetration rates, except for an almost negligible improvement near the density at capacity, i.e. the value of $\rho$ at which $q$ is maximum. Consequently, the solution to~\eqref{eq:first_order} will be virtually the same regardless of $p$. This suggests that the considered driver-assist controls do not impact on the macroscopic flow of vehicles. Nevertheless, this does not mean that they do not produce any aggregate effect on the traffic. Owing to the results of Section~\ref{sect:local_traffic}, a suitable level at which to study such an effect is the kinetic/statistical one.

\begin{remark}[Second order hydrodynamics] \label{rem:second}
One may wonder whether the substantial noninfluence of the driver-assist controls at the macroscopic scale could be corrected by a higher order hydrodynamical model. Actually, the physics of the interactions~\eqref{eq:binary.controlled} induces necessarily a first order local-equilibrium-based hydrodynamics, because the interactions conserve only the mass of the vehicles. In order to obtain a \textit{second order} macroscopic model, i.e. one made of two equations for the first two moments~\eqref{eq:rho},~\eqref{eq:h} of the kinetic distribution function, a classical option is to consider a \textit{monokinetic closure} of~\eqref{eq:boltzmann}, which amounts to forcing
\begin{equation}
	f(t,\,x,\,s)=\rho(t,\,x)\delta(s-h(t,\,x)).
	\label{eq:f.monokin}
\end{equation}
This corresponds to postulating that the variance of the headway distribution is locally zero, namely that all vehicles keep locally the same (unknown) distance $h$ from their leading vehicles. As a matter of fact,~\eqref{eq:f.monokin} is the simplest distribution depending on the two hydrodynamic parameters $\rho$, $h$. Plugging~\eqref{eq:f.monokin} into~\eqref{eq:boltzmann} and recalling~\eqref{eq:binary.controlled} yields then, for $\varphi(s)=1,\,s$,
\begin{equation}
	\begin{cases}
		\partial_t\rho+\partial_x\dfrac{\rho h}{a+h}=0 \\
		\partial_t(\rho h)+\partial_x\dfrac{\rho h^2}{a+h}
			=\dfrac{p\mu\rho^2}{2(\nu+1)}(s_d(\rho)-h),
	\end{cases}
	\label{eq:second.order}
\end{equation}
which is a \textit{pressureless} hydrodynamic model. Nevertheless,~\eqref{eq:f.monokin} is only an \textit{ansatz}, which, unlike~\eqref{eq:f.first_order}, is \textit{not} justified by the physics of the microscopic interactions. Therefore, as we will see in the numerical tests of Section~\ref{sect:second.num}, we \textit{cannot} expect the macroscopic description~\eqref{eq:second.order} to be consistent with the actual aggregate dynamics produced by the kinetic model~\eqref{eq:binary.controlled},~\eqref{eq:boltzmann}.

In conclusion, the physically correct hydrodynamics is the first order one, cf.~\eqref{eq:first_order},~\eqref{eq:q.s}, together with the physiological noninfluence of the considered driver-assist controls on the macroscopic traffic.
\end{remark}

\section{Numerical tests}
\label{sect:numerics}
In this section, we propose several numerical tests to compare the inhomogeneous Boltzmann-type model~\eqref{eq:binary.controlled}-\eqref{eq:boltzmann} with the first and second order hydrodynamic models~\eqref{eq:first_order}-\eqref{eq:q.v},~\eqref{eq:second.order} derived therefrom.

For the numerical solution of the Boltzmann-type kinetic model in the quasi-invariant scaling~\eqref{eq:quasi-invariant_scaling} we adopt a Monte Carlo approach, see~\cite{pareschi2001ESAIMP,pareschi2013BOOK} and also~\cite{dimarco2014AN} for a survey on recent advances. For the sake of completeness, we briefly sketch here the basic ideas of the method. We consider the scaled Boltzmann-type equation~\eqref{eq:boltzmann.scaled} in strong form:
\begin{equation}
	\partial_tf(t,\,x,\,s)+\frac{s}{a+s}\partial_xf(t,\,x,\,s)=\frac{1}{\epsilon}Q(f,\,f)(t,\,x,\,s),
	\label{eq:boltz_strong}
\end{equation}
where $Q$ is the ``collisional'' operator:
\begin{equation}
	Q(f,\,f)(t,\,x,\,s)=\E_{\Theta}\ave*{\int_{\R_+}\left(\frac{1}{\pr{J}}f(t,\,x,\,\pr{s})f(t,\,x,\,\pr{s_\ast})-f(t,\,x,\,s)f(t,\,x,\,s_\ast)\right)\,ds_\ast},
	\label{eq:Q}
\end{equation}
where $\pr{s}$, $\pr{s_\ast}$ are the pre-interaction headways that produce the post-interaction headways $s$, $s_\ast$ according to the interaction rules~\eqref{eq:binary.controlled} and $\pr{J}$ is the Jacobian of the transformation from $(\pr{s},\,\pr{s_\ast})$ to $(s,\,s_\ast)$. As already mentioned in Section~\ref{sect:boltzmann.hydro}, also at the numerical level a classical strategy to solve the equation consists in a splitting of the collision and the transport parts. The core idea is that for $\epsilon\ll 1$ one transports the local equilibrium distribution, which is quickly reached in time in each point $x$:
$$ \partial_tf=\frac{1}{\epsilon}Q(f,\,f) \quad \to \quad \partial_tf+\frac{s}{a+s}\partial_xf=0. $$
In particular, in order to find numerically the local equilibrium (first step) we use the classical Nanbu algorithm for Maxwellian molecules, see e.g.~\cite{bobylev2000PRE}. We refer the interested reader also to~\cite{pareschi2019JNS,tosin2019MMS} for a more detailed discussion on this method with direct applications to multi-agent systems.

On the other hand, for the numerical solution of the conservation law~\eqref{eq:first_order} with the neither convex nor concave flux~\eqref{eq:q.v} we consider a finite volume weighted essentially non-oscillatory (WENO) scheme, see~\cite{shu2009SIREV}, with Godunov numerical flux. This type of problem is hard to tackle at the numerical level with high order schemes, because the convergence to the correct entropy solution may fail in the absence of strict convexity or concavity of the flux. In order to overcome this difficulty, we follow the approach described in~\cite{qiu2008SISC}, where a modification of the WENO method is suggested based on suitable monotone schemes which maintain high order accuracy far from the discontinuities of the solution. The resulting scheme is high order accurate in the regions where the solution is smooth and uses a $O(1)$ indicator near the non-convex/concave discontinuity regions. Specifically, we introduce a uniform mesh with given mesh size $\Delta{x}>0$ in the space domain and, at each time, we choose the time step $\Delta{t}>0$ according to the CFL condition
$$ \frac{\Delta{t}}{\Delta{x}}\max_{x}\abs{q'(\rho(t,\,x))}=1. $$

\subsection{First order hydrodynamics}
In the following, we consider~\eqref{eq:boltz_strong}-\eqref{eq:Q} for $(x,\,s)\in [-4,\,4]\times [0,\,20]$ with periodic boundary conditions on the space variable $x$. As initial condition, we prescribe the following distribution function:
$$ f_0(x,\,s)= 
	\begin{cases}
		0.02 & \text{if } (x,\,s)\in [-2,\,0]\times [0,\,10] \\
		0.03 & \text{if } (x,\,s)\in [0,\,2]\times [0,\,10] \\
		0 & \text{otherwise}
	\end{cases}
$$
to which there corresponds the initial vehicle density
$$ \rho_0(x)=\int_0^{20}f_0(x,\,s)\,ds=
	\begin{cases}
		0.2 & \text{if } x\in [-2,\,0) \\
		0.3 & \text{if } x\in [0,\,2] \\
		0 & \text{otherwise}.
	\end{cases}
$$

\begin{figure}[!t]
\centering
\includegraphics[scale=0.5]{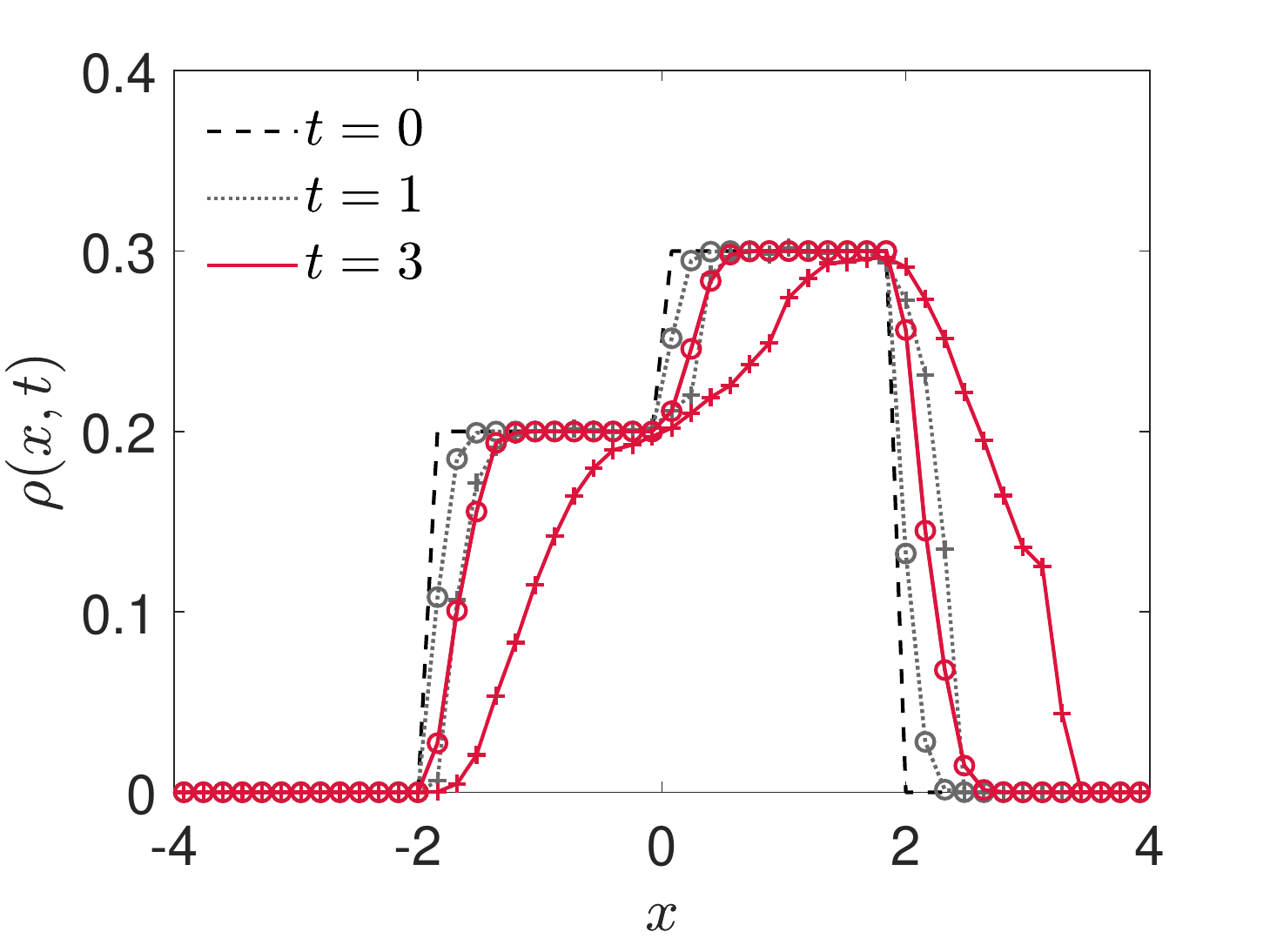}
\includegraphics[scale=0.5]{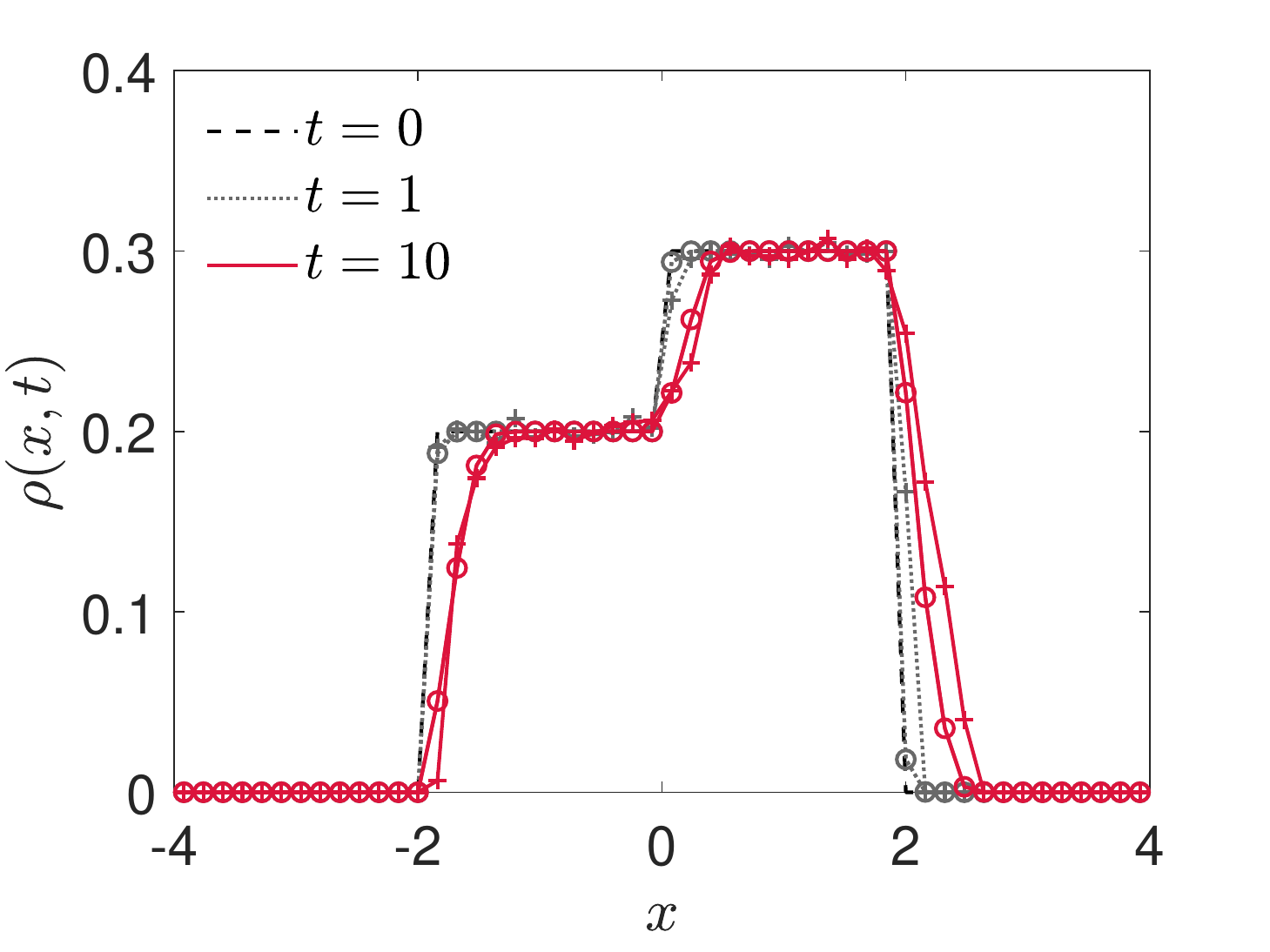} \\
\includegraphics[scale=0.5]{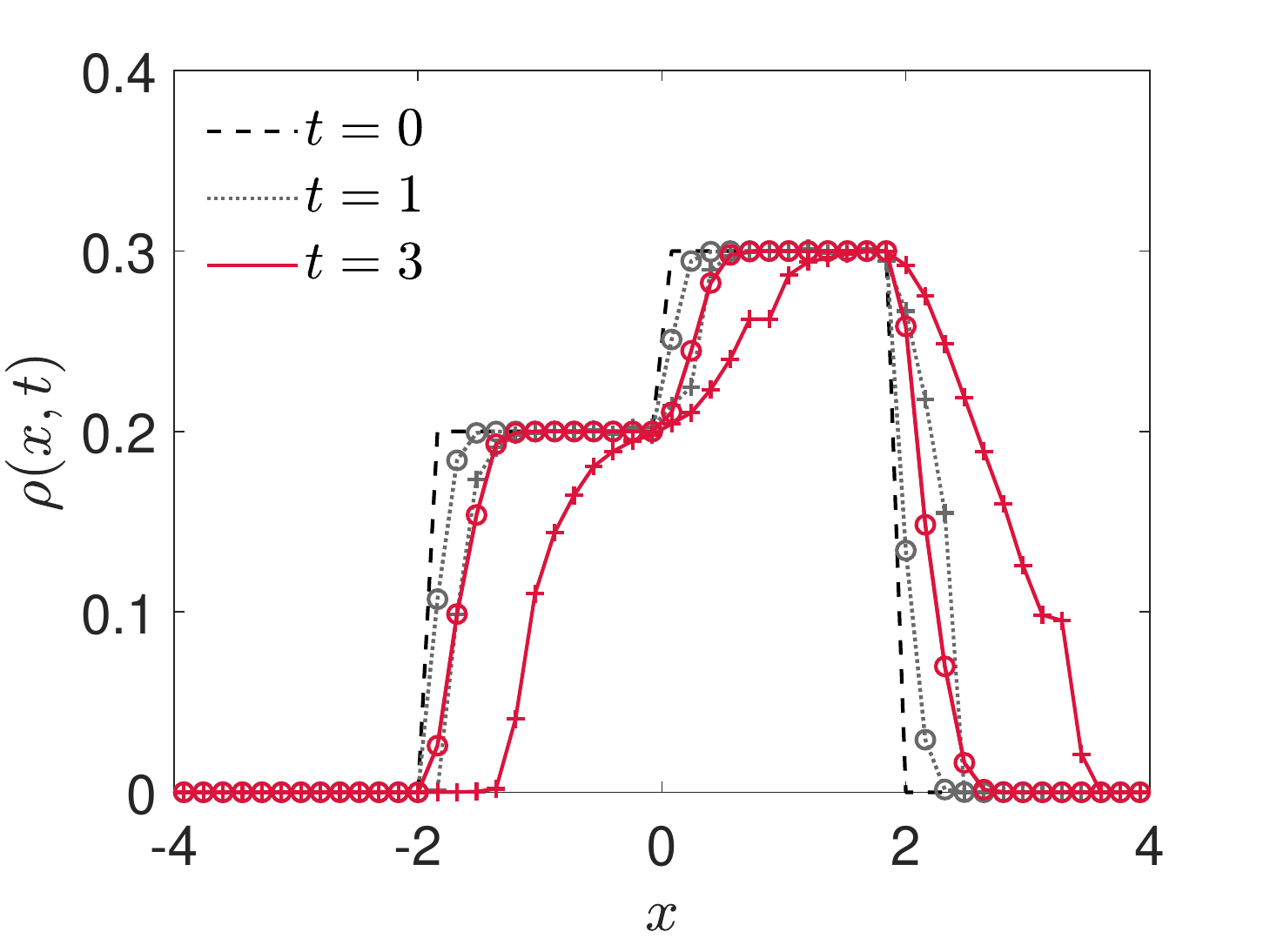}
\includegraphics[scale=0.5]{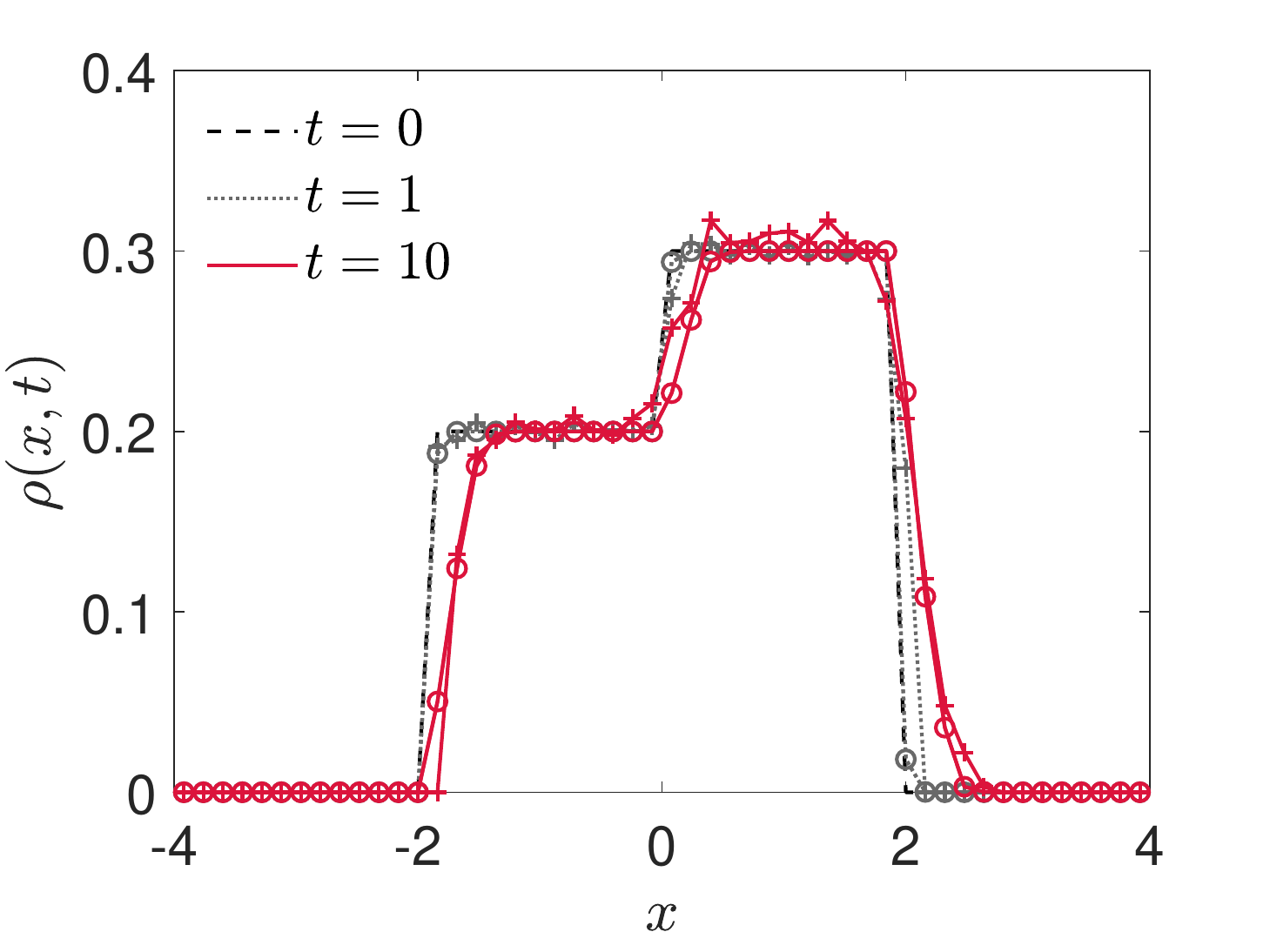}
\caption{The density $\rho$ at successive times reconstructed from the kinetic model (marker $+$) and computed from the first order macroscopic model (marker $\circ$). \textbf{Top row}: $p=5\%$; \textbf{bottom row}: $p=50\%$. \textbf{Left column}: $\epsilon=10^{-2}$; \textbf{right column}: $\epsilon=10^{-4}$.}
\label{fig:first_fig1}
\end{figure}

In Figure~\ref{fig:first_fig1}, we compare the evolution of the density $\rho$ reconstructed from the numerical solution of the kinetic model (marker $+$) with the one obtained from the first order hydrodynamic model (marker $\circ$). We consider, in particular, the cases $\epsilon=10^{-2}$ and $\epsilon=10^{-4}$, which, under the quasi-invariant scaling~\eqref{eq:quasi-invariant_scaling}, imply $a=10$ and $a=100$, respectively. We clearly observe that for vanishing $\epsilon$ the two solutions are consistent with each other, as expected from the theory. Conversely, if $\epsilon$ is not small enough they may greatly differ from one another, because the actual regime of the kinetic model is far from the hydrodynamic regime.

\begin{figure}[!t]
\centering
\includegraphics[scale=0.5]{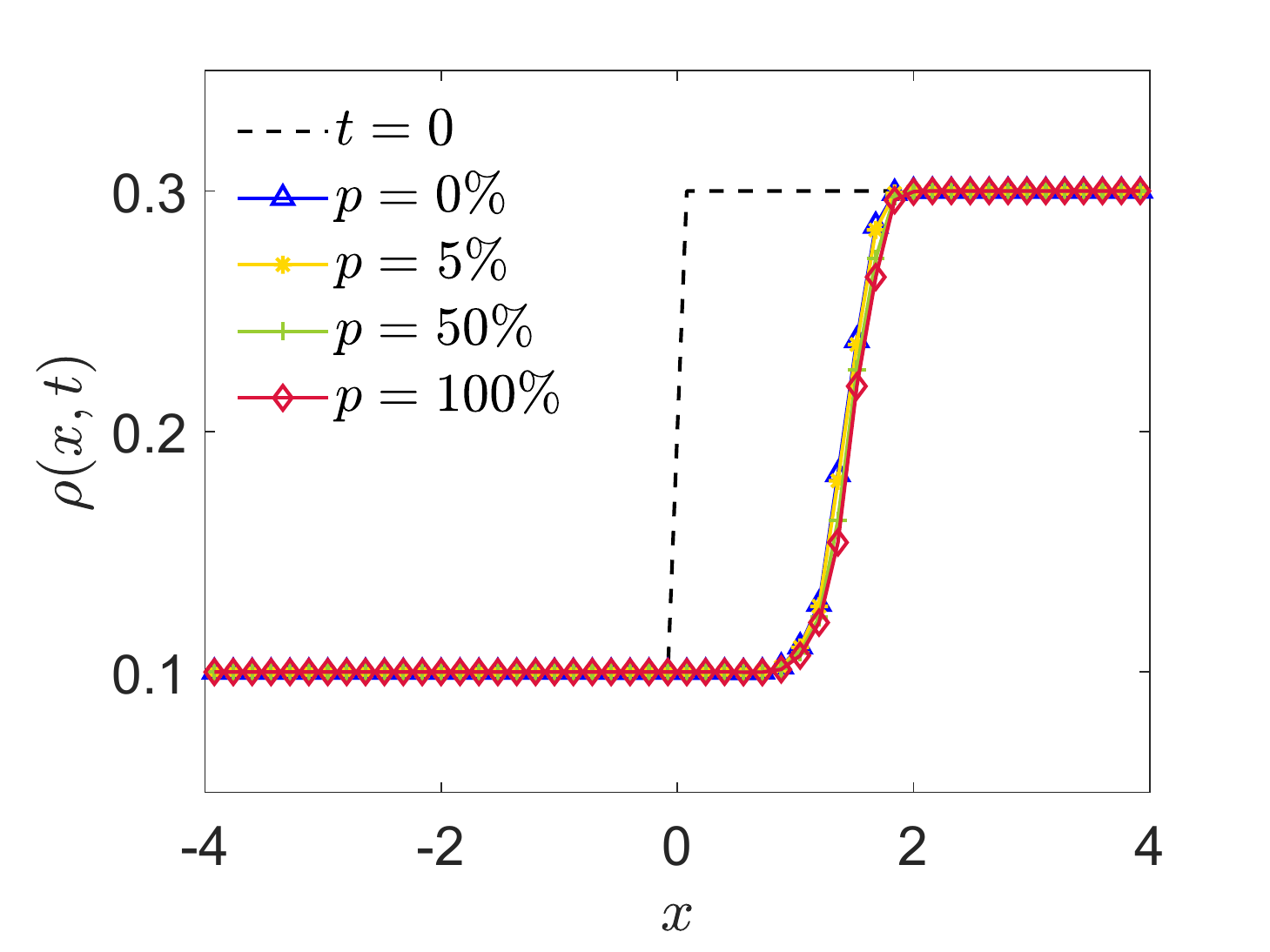}
\includegraphics[scale=0.5]{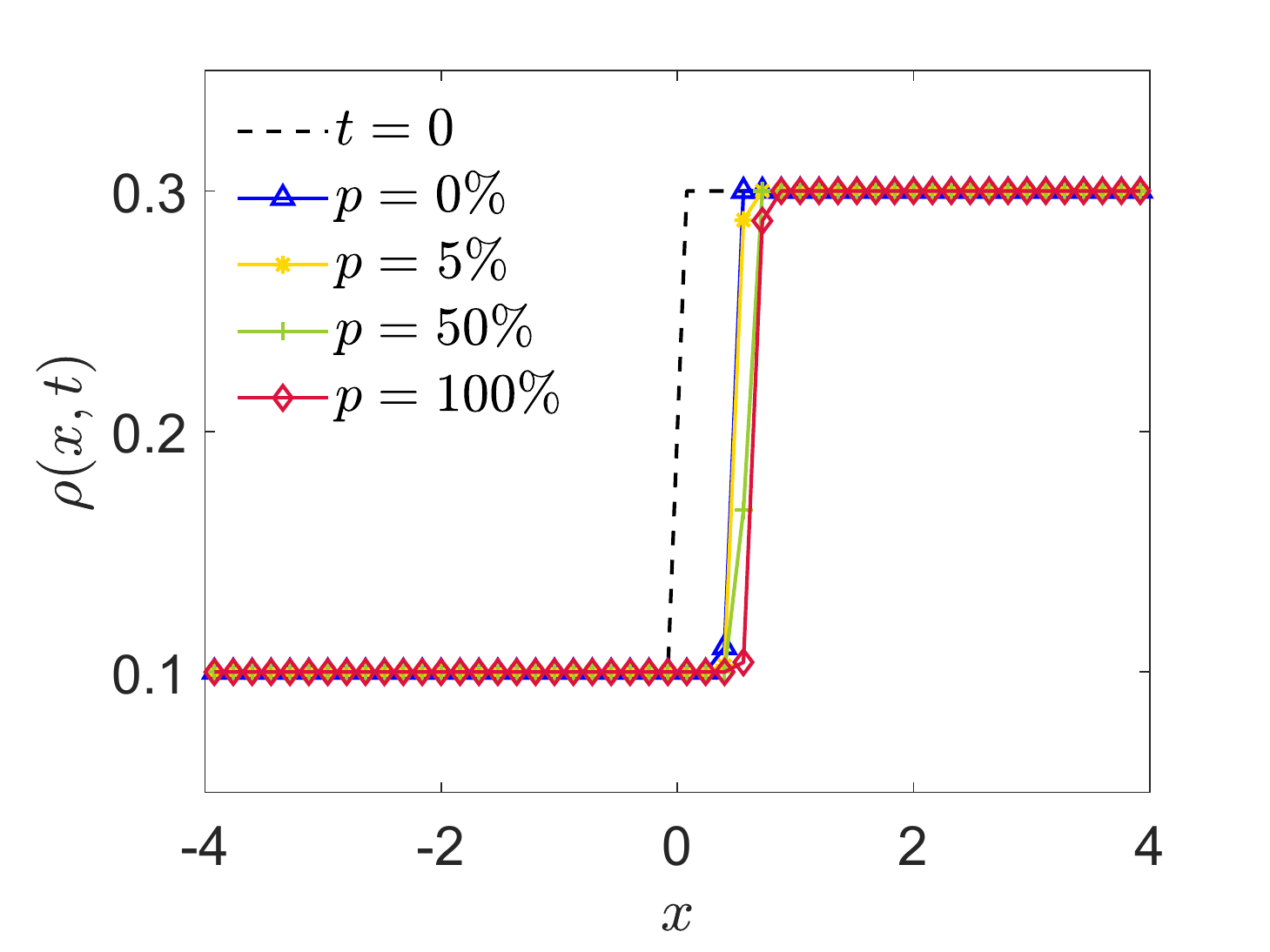}
\caption{Numerical solution at time $t=3$ of the first order macroscopic model for several choices of the penetration rate $p$. \textbf{Left}: $a=1$; \textbf{right}: $a=10$.}
\label{fig:first_fig2}
\end{figure}

In Figure~\ref{fig:first_fig2}, we show the numerical solution of the first order macroscopic model at time $t=3$ for $a=1,\,10$ and for several values of the penetration rate $p$. As predicted by the theory, a very small influence of $p$ can be detected on the macroscopic flow, see also Figure~\ref{fig:funddiag} in Section~\ref{sect:hydro_lim}.

\begin{remark}
The Monte Carlo solution of the kinetic model has been obtained using $10^5$ particles. At each time step, the macroscopic density has been reconstructed using $50$ grid points in the space interval $[-4,\,4]$ and integrating the numerical kinetic distribution function with respect to $s$ in the interval $[0,\,20]$ with $100$ grid points.
\end{remark}

\subsection{Second order hydrodynamics}
\label{sect:second.num}
We turn now to the second order hydrodynamic model~\eqref{eq:second.order} derived under the monokinetic \textit{ansatz}~\eqref{eq:f.monokin}. Pressureless systems of balance laws have been quite extensively investigated in recent years at both the analytical and the numerical levels, see e.g.~\cite{bouchut1994BOOK,bouchut2003SINUM}. These systems are weakly hyperbolic, thus, in the absence of source terms, the vacuum (i.e., $\rho=0$ in some space regions) may appear in finite time, thereby producing a propagation of the momentum at infinite speed. In principle, with standard finite volume methods a vanishing time step is necessary to guarantee the stability of the numerical scheme. In order to bypass such structural problems, in the following we will confine our numerical tests to short enough time horizons, however sufficient for a major comparison of the kinetic and the macroscopic solutions. In more detail, we solve numerically system~\eqref{eq:second.order} by means of a standard finite volume WENO scheme and a splitting of the transport and the relaxation dynamics. We consider the space domain $[-1,\,1]$, that we discretise by means of $50$ grid points. Moreover, we set $\Delta{t}=\Delta{x}/2$, which corresponds to a CFL number equal to $0.5$. On the other hand, for the numerical solution of the kinetic model we apply a Monte Carlo algorithm with $10^5$ particles in the space of the microscopic states $(x,\,s)\in [-1,\,1]\times [0,\,20]$. Then we reconstruct the hydrodynamic parameters on the same space mesh of the macroscopic model and integrating numerically the kinetic distribution function with respect to $s\in [0,\,20]$ with $10^3$ grid points.

We prescribe the following initial condition:
$$ f_0(x,\,s)=
	\begin{cases}
		0.16 & \text{if } (x,\,s)\in [-1,\,0]\times [0,\,5] \\
		0.02 & \text{if } (x,\,s)\in [0,\,1]\times [0,\,10] \\
		0 & \text{otherwise},
	\end{cases}
$$
which induces, at the macroscopic level,
\begin{align*}
	\rho_0(x) &= \int_0^{20}f_0(x,\,s)\,ds=
	\begin{cases}
		0.8 & \text{if } x\in [-1,\,0] \\
		0.2 & \text{if } x\in (0,\,1],
	\end{cases}
	\\
	h_0(x) &= \frac{1}{\rho_0(x)}\int_0^{20}sf_0(x,\,s)\,ds=
	\begin{cases}
		2.5 & \text{if } x\in [-1,\,0] \\
		5 & \text{if } x\in (0,\,1].
	\end{cases}
\end{align*}

\begin{figure}[!t]
\centering
\includegraphics[scale=0.5]{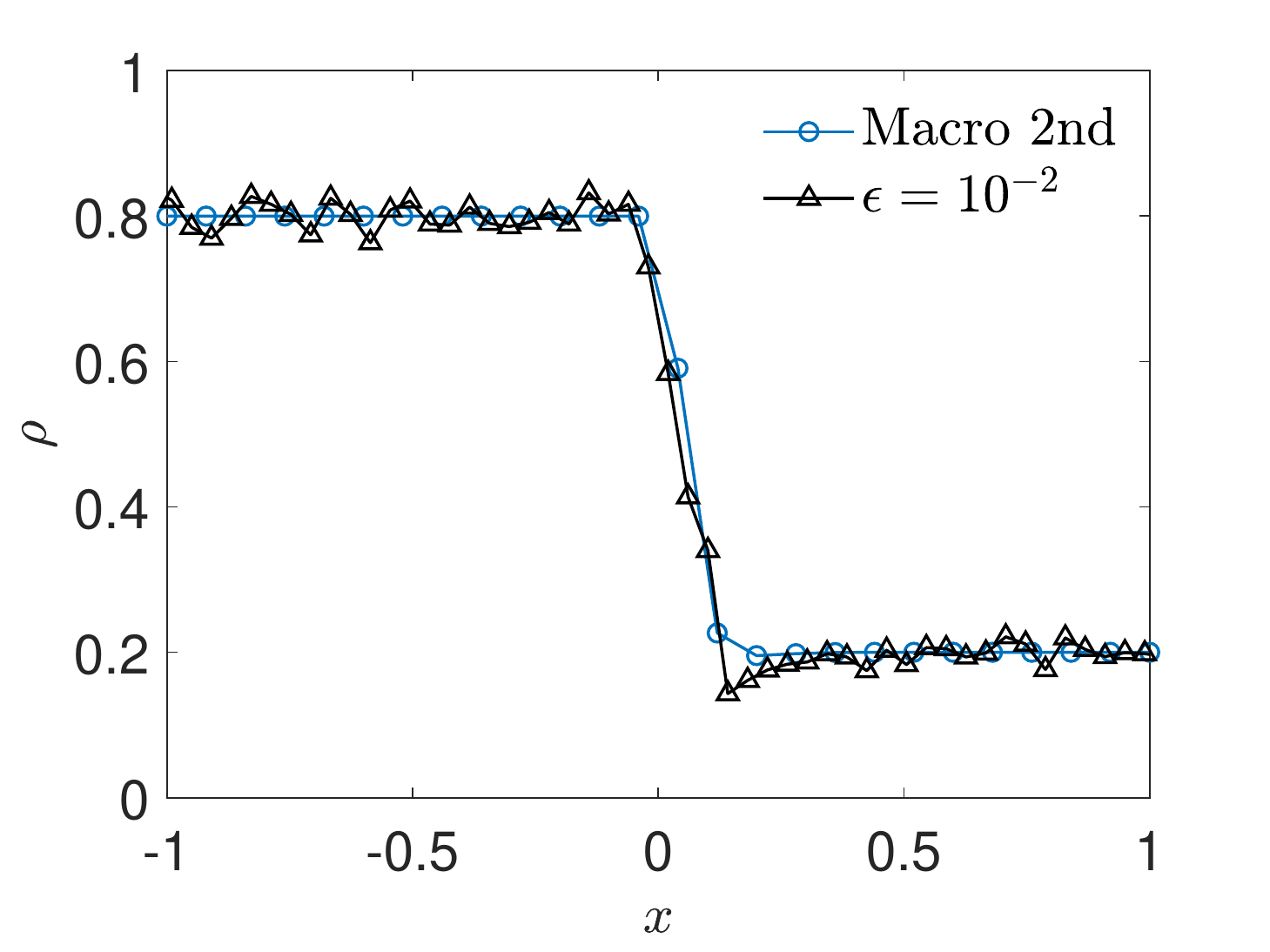}
\includegraphics[scale=0.5]{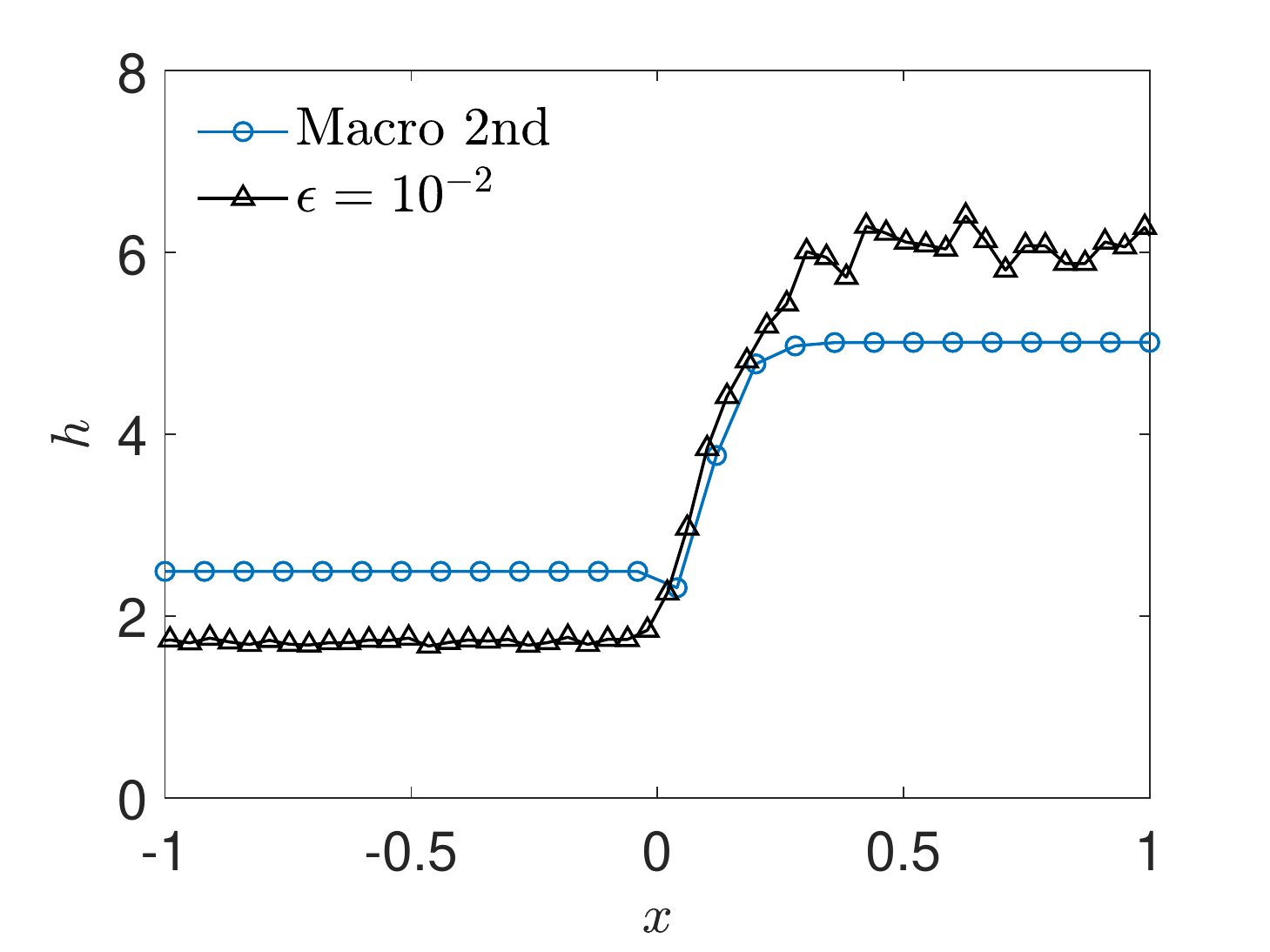} \\
\includegraphics[scale=0.5]{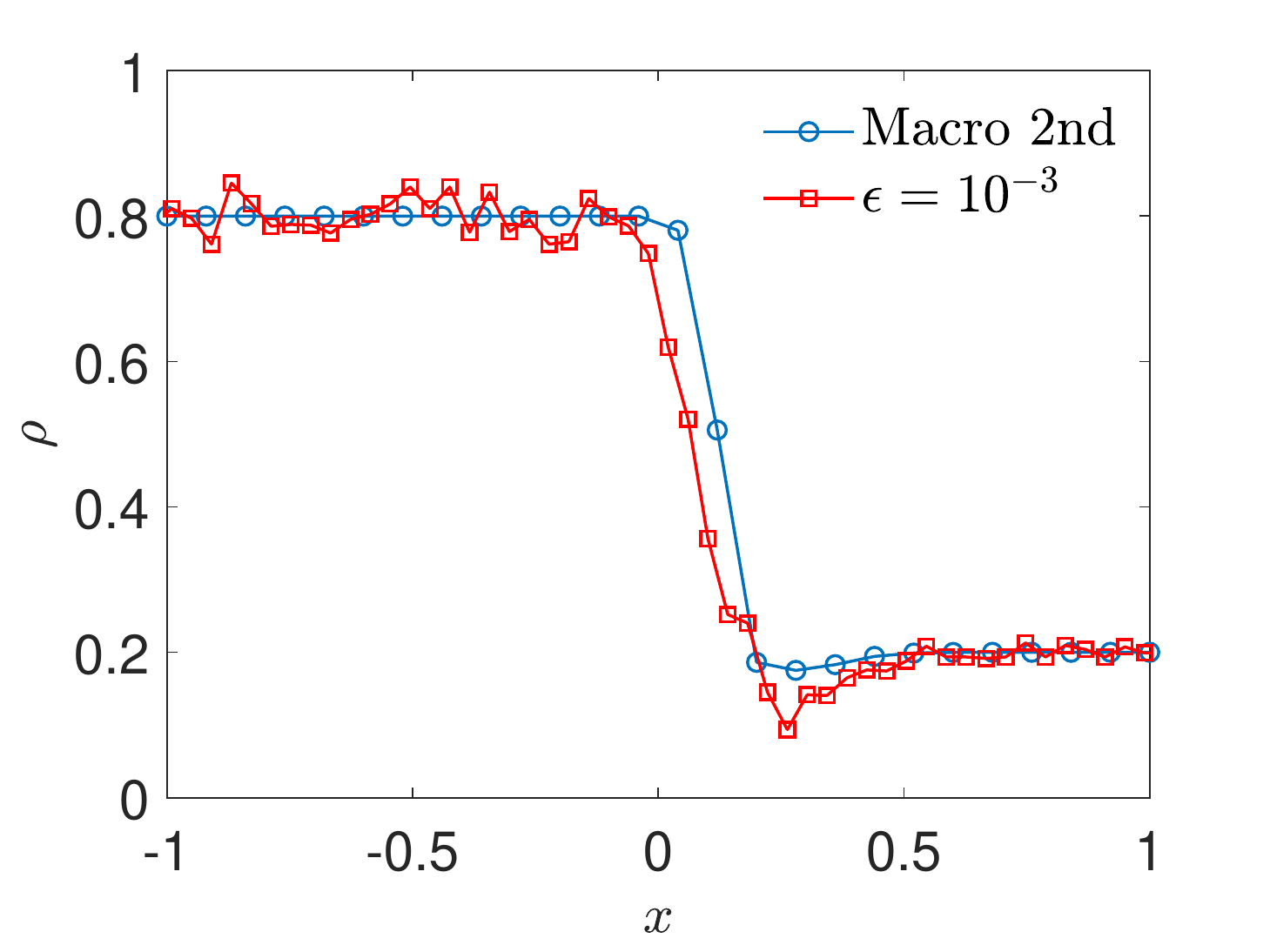}
\includegraphics[scale=0.5]{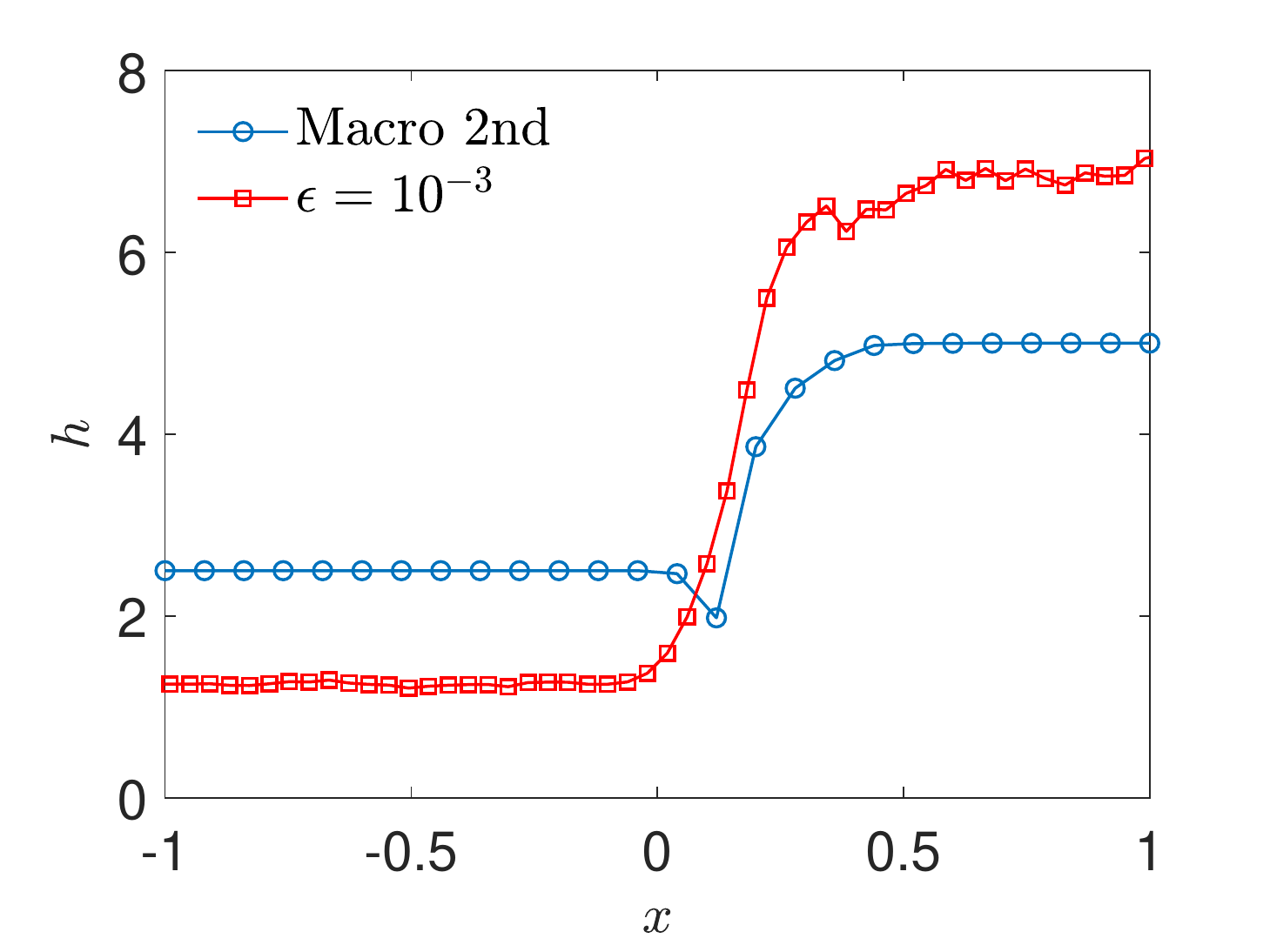}
\caption{Comparison of the kinetic model and the second order macroscopic model at time $t=2$ with the scaling $\epsilon=10^{-2}$ (\textbf{top row}) and $\epsilon=10^{-3}$ (\textbf{bottom row}). \textbf{Left column}: density $\rho(t,\,x)$; \textbf{right column}: mean headway $h(t,\,x)$.}
\label{fig:second_fig1}
\end{figure}

In Figure~\ref{fig:second_fig1} we compare the evolution of $\rho$, $h$ predicted by the second order hydrodynamic model~\eqref{eq:second.order} and the ones reconstructed from the kinetic model. In particular, we use the following values of the parameters: $a=30$, $\nu=\frac{1}{\epsilon}$, $\Delta{t}=\epsilon$, $\sigma^2=\Var(\eta)=0$ with $\epsilon=10^{-2}$ and $\epsilon=10^{-3}$, in order mimic ideally a pressureless hydrodynamic regime also at the level of the microscopic interactions, cf.~\eqref{eq:binary.controlled}. The numerical results confirm the expectation anticipated by the theory, cf. Remark~\ref{rem:second}: the second order hydrodynamic model~\eqref{eq:second.order} is only a rough approximation of the kinetic model quite independently of $\epsilon$, because the \textit{ansatz}~\eqref{eq:f.monokin} is not supported by the microscopic physics of the system.

\section{Conclusions}
\label{sect:conclusions}
In this paper we have investigated the impact of automatic cruise control systems at various traffic scales using a model-based Boltzmann-type kinetic theory approach. Specifically, we have provided a theoretical evidence of the fact that the controls aimed at optimising the headway and aligning the speed of the vehicles are effective on a local aggregate traffic scale but they virtually do not impact on the average fluid regime. In other words, the local statistical properties of the flow of vehicles, such as e.g. the variance of the headway and speed distributions, are optimised whereas no relevant differences are observed in the overall flux/throughput compared to a traffic stream composed of fully human-manned vehicles. Such a result supports nicely recent experimental findings~\cite{stern2018TRC}, which pointed out that the aforesaid types of controls are able to dampen local traffic instabilities, for instance stop-and-go waves ascribable to the heterogeneous distributions of headways and speeds, while being essentially non-intrusive in the macroscopic flow of the vehicles. This in turn demonstrates that those driver-assist controls can potentially mitigate road risk factors due to heterogeneous headways and speeds without affecting the efficiency of the global traffic flow.

The model-based approach proposed in this paper suggests that the methods of the kinetic theory are valid tools to unravel the multiscale aspects of driver-assist control strategies in a formally rigorous, organic and consistent way. We believe that these methods can be further developed to tackle more ambitious issues, such as the optimal design of vehicle-wise driver-assist algorithms able to induce prescribed desired effects at higher traffic scales.

\section*{Acknowledgements}
This research was partially supported by the Italian Ministry of Education, University and Research (MIUR) through the ``Dipartimenti di Eccellenza'' Programme (2018-2022) -- Department of Mathematical Sciences ``G. L. Lagrange'', Politecnico di Torino (CUP:E11G18000350001) and Department of Mathematics ``F. Casorati'', University of Pavia; and through the PRIN 2017 project (No. 2017KKJP4X) ``Innovative numerical methods for evolutionary partial differential equations and applications''.

This work is also part of the activities of the Starting Grant ``Attracting Excellent Professors'' funded by ``Compagnia di San Paolo'' (Torino) and promoted by Politecnico di Torino.

B.P. acknowledges the support of the National Science Foundation under the CPS Synergy Grant No. CNS-1837481.

A.T. and M.Z. are members of GNFM (Gruppo Nazionale per la Fisica Matematica) of INdAM (Istituto Nazionale di Alta Matematica), Italy.

\bibliographystyle{plain}
\bibliography{PbTaZm-hydroFTL_traffic}

\end{document}